\documentclass{aiml10}

\usepackage{aiml10macro}
\usepackage{proof}
\usepackage{bussproofs}
\usepackage{latexsym}
\usepackage{amsmath}
\usepackage{amssymb}
\usepackage{url}

\newcommand{\lint}{\textrm{Int}}
\newcommand{\dint}{\textrm{DInt}}
\newcommand{\intk}{\textrm{IntK}}
\newcommand{\biint}{\hbox{BiInt}}
\newcommand{\kt}{\hbox{Kt}}
\newcommand{\bikt}{\hbox{BiKt}}
\newcommand{\lbikt}{\mathrm{\bf LBiKt}}
\newcommand{\lbikte}{\mathrm{\bf LBiKtE}}
\newcommand{\dbikt}{\mathrm{\bf DBiKt}}

\newcommand{\dbikts}{\mathrm{\bf DBiKt_1}}
\newcommand{\di}{\mathrm{\bf DInt}}
\newcommand{\dik}{\mathrm{\bf DIntK}}
\newcommand{\dbi}{\mathrm{\bf DBInt}}
\newcommand{\struct}{\triangleright}
\newcommand{\derive}[2]{\vdash_{#2} #1:}
\newcommand{\impl}{\rightarrow}
\newcommand{\dimpl}{{\mbox{$\;-\!\!\!<\;$}}}
\newcommand{\upcontext}[1]{\widehat{#1}}
\newcommand{\toplevel}[1]{\{\!| #1 |\!\}}
\newcommand{\boxf}{\square}
\newcommand{\diaf}{\lozenge}
\newcommand{\boxp}{\blacksquare}
\newcommand{\diap}{\blacklozenge}
\newcommand{\dneg}{\sim}
\newcommand{\circwhite}[1]{\circ ( #1 )}
\newcommand{\circblack}[1]{\bullet ( #1 )}
\newcommand{\circwhiteshort}[1]{\circ #1}
\newcommand{\circblackshort}[1]{\bullet #1}
\newcommand{\deriveshallow}[2]{\derive{#1}{\lbikt} #2}
\newcommand{\derivedeep}[2]{\derive{#1}{\dbikt} #2}
\newcommand{\mysmall}[1]{\mbox{\small $#1$}}

\renewcommand{\qed}{$\dashv$}

\def\lastname{Gor\'e, Postniece and Tiu}

\begin{document}

\begin{frontmatter}
  \title{Cut-elimination and Proof Search for Bi-Intuitionistic Tense Logic}
  \author{Rajeev Gor\'e, Linda Postniece \and Alwen Tiu}
  \address{Logic and Computation Group, School of Computer Science \\
The Australian National University \\
\url{{Rajeev.Gore, Linda.Postniece, Alwen.Tiu}@anu.edu.au}}

\begin{abstract}
  We consider an extension of bi-intuitionistic logic with the
  traditional modalities $\diaf$, $\boxf$, $\diap$ and $\boxp$
  from tense logic \kt{}. Proof theoretically, this extension is obtained
  simply by extending an existing sequent calculus for bi-intuitionistic logic
  with typical inference rules for the modalities used in display logics. 
  As it turns out, the resulting calculus,
  $\lbikt$, seems to be more basic than most intuitionistic tense or modal logics
  considered in the literature, in particular, those studied by Ewald and Simpson,
  as it does not assume any {\em a priori} relationship between the modal
  operators $\diaf$ and $\boxf$. 
  We recover Ewald's intuitionistic tense logic and Simpson's 
  intuitionistic modal logic by modularly extending $\lbikt$ with 
  additional structural rules. 
  The calculus $\lbikt$ is formulated in a variant of display calculus, 
  using a form of sequents called nested sequents. Cut elimination is proved
  for $\lbikt$, using a technique similar to that used in display calculi. 
  As in display calculi, the inference rules of $\lbikt$ are ``shallow'' rules, 
  in the sense that they act on top-level formulae in a nested sequent. 
  The calculus $\lbikt$ is ill-suited for 
  backward proof search due to the presence of certain structural rules
  called ``display postulates'' and the contraction rules on arbitrary
  structures. We show that these structural rules can be made redundant
  in another calculus, $\dbikt$, which uses deep inference, allowing
  one to apply inference rules at an arbitrary depth in a nested sequent. 
  We prove the equivalence between $\lbikt$ and $\dbikt$ and 
  outline a proof search strategy for $\dbikt$. 
  We also give a Kripke semantics and prove that $\lbikt$ is sound with respect
  to the semantics, but completeness is still an open problem. 
  We then discuss various extensions of $\lbikt.$ 

\end{abstract}


\end{frontmatter}

\section{Introduction}

Intuitionistic logic \lint{} forms a rigorous foundation for many
areas of Computer Science via its constructive interpretation and via
the Curry-Howard isomorphism between natural deduction proofs
and well-typed terms in the $\lambda$-calculus. Central to
both concerns are syntactic proof calculi with 
cut-elimination and  backwards
proof-search for finding derivations automatically.

In traditional intuitionistic logic, the connectives $\impl$ and $\land$ form an
adjoint pair in that $(A \land B) \impl C$ is valid iff $A \impl (B \impl C)$ 
is valid iff $B \impl (A \impl C)$ is valid.  
Rauszer~\cite{Rauszer80Phd} obtained BiInt by extending
Int with a binary connective $\dimpl$ called ``exclusion'' 
which is adjoint to $\lor$ 
in that $A \impl (B \lor C)$ is valid iff $(A \dimpl B) \impl C$ 
is valid iff $(A \dimpl C) \impl B$ is valid. 
Crolard~\cite{crolard2004} showed that \biint{} has a
computational interpretation in terms of continuation passing style
semantics.  Uustalu and Pinto recently showed that Rauszer's
sequent calculus~\cite{Rauszer74StudiaLogica} and
Crolard's extensions of it fail cut-elimination, 
but a nested sequent calculus with 
cut-elimination~\cite{gorepostniecetiu2008} and
a labelled sequent calculus~\cite{DBLP:conf/tableaux/PintoU09} with
cut-free-completeness have been found for \biint{}.

The literature on Intuitionistic Modal/Tense Logics (IM/TLs) 
is vast~\cite{ewald1986,simpson1994} and typically uses Hilbert calculi with
algebraic, topological or relational semantics.  We omit
details since our interest is primarily proof-theoretic.
Sequent and natural deduction calculi for IMLs are
rarer~\cite{DBLP:journals/logcom/Masini93,DBLP:journals/sLogica/AmatiP94,DBLP:journals/entcs/PfenningW95,collinson-hilken-rydeheard-intuitionistic-modal-sequent,DBLP:journals/jacm/DaviesP01,kakutani-inml2007,galmiche-salhi-intuitionistic-hybrid-modal-logic}.
Extending them with ``converse'' modalities like $\diap$
and $\boxp$ causes cut-elimination to fail as it does for classical
modal logic S5 where $\diaf$ is a
self-converse. Labels~\cite{mints-labelled-sfive,simpson1994,DBLP:conf/lics/VIICHP04} 
can help but are not purely proof-theoretic
since they encode the Kripke semantics.

The closest to our work is that of 
Sadrzadeh and Dyckhoff \cite{DBLP:journals/entcs/SadrzadehD09} who
give a cut-free sequent calculus using deep inference for a logic with
an adjoint pair of modalities $(\diap, \boxf)$ plus only $\land$,
$\lor$, $\top$ and $\bot$. 
As all their connectives are
``monotonic'', cut-elimination presents no difficulties.

Let \bikt{} be the bi-intuitionistic tense logic obtained by extending
\biint{} with two pairs of adjoint modalities $(\diaf, \boxp)$ and
$(\diap, \boxf)$, with no explicit relationship between the modalities
of the same colour, namely, $(\diaf, \boxf)$ and $(\diap, \boxp)$.
The modalities form an adjunction as follows: 
$A \impl \boxf B$ iff $\diap A \impl B$
and
$A \impl \boxp B$ iff $\diaf A \impl B.$ 

Our shallow inference calculus $\lbikt$ is a merger of two sub-calculi
for \biint{} and \kt{} derived from Belnap's inherently modular
display logic. $\lbikt$ has syntactic cut-elimination,
but is ill-suited for backward proof search. Our
deep inference calculus $\dbikt$ is complete with respect to the
cut-free fragment of $\lbikt$ and is more amenable to proof search
as it contains no display postulates and contraction rules.

To complete the picture, we also give a Kripke semantics for \bikt{}
based upon three relations $\leq$, $R_{\diaf}$ and
$R_{\boxf}$.
The logic \bikt{} enjoys various desirable properties:
\begin{itemize}
\item{\bf Conservativity:} it is a conservative extension of 
  intuitionistic logic \lint{}, dual intuitionistic logic \dint{},
  and bi-intuitionistic logic \biint{};

\item{\bf Classical Collapse:} it collapses to classical tense logic
  by the addition of four structural rules;
\item{\bf Disjunction Property:} If $A \lor B$ is a theorem not
  containing $\dimpl$ then $A$ is a theorem or $B$ is a theorem;

\item{\bf Dual Disjunction Property:} If $A \land B$ is a
  counter-theorem not containing $\impl$ then so is $A$ 
  or $B$;

\item{\bf Independent $\diaf$ and $\boxf$:} there is no {\em a priori}
  relationship between these connectives.
\end{itemize}
The independence of $\diaf$ and $\boxf$ is a departure
from traditional intuitionistic tense or modal logics, e.g.,
those considered by Ewald~\cite{ewald1986} and Simpson~\cite{simpson1994}. 
Both Ewald and Simpson allow a form of interdependency between $\diaf$ and $\boxf$, expressed
as the axiom $(\diaf A \impl \boxf B) \impl \boxf (A \impl B)$, which is not
derivable in $\lbikt.$ However, we shall see in Section~\ref{sec:ext} that we
can recover Ewald's intuitionistic tense logic and Simpson's intuitionistic
modal logic by extending $\lbikt$ with two structural rules.

\section{Nested Sequents}\label{sec:nested}

\begin{figure}[t]
{\small
$$
\begin{array}{lll@{\qquad\qquad}lll}
\tau^-(A) & = & A 
  & \tau^+(A) & = & A \\
\tau^-(X, Y) & = & \tau^-(X) \land \tau^-(Y) 
  & \tau^+(X, Y) & = & \tau^+(X) \lor \tau^+(Y) \\
\tau^-(X \struct Y) & = & \tau^-(X) \dimpl \tau^+(Y) 
  & \tau^+(X \struct Y) & = & \tau^-(X) \impl \tau^+(Y) \\
\tau^-(\circ X) & = & \diaf \tau^-(X) 
  & \tau^+(\circ X) & = & \boxf \tau^+(X) \\
\tau^-(\bullet X) & = & \diap \tau^-(X) 
  & \tau^+(\bullet X) & = & \boxp \tau^+(X) \\
\tau(X \struct Y) & = & \tau^-(X) \impl \tau^+(Y) 
\end{array}
$$
}
\caption{Formula Translation of Nested Sequents}
\label{fig:interpretation}
\end{figure}

The formulae of $\bikt$ are built from a set $Atoms$ of
atomic formulae via the grammar below, with $p \in Atoms$:
$$
  A ::= p \mid \top \mid \bot \mid A \impl A \mid A \dimpl A 
             \mid A \land A \mid A \lor A
             \mid \boxf A \mid \diaf A \mid \boxp A \mid \diap A.
$$
A structure is defined by the following grammar, where $A$ is a $\bikt$ formula:
$$X := \emptyset \mid A \mid (X, X) \mid X \struct X \mid \circwhiteshort X \mid \circblackshort X.$$

The structural connective ``,'' is associative and commutative
and $\emptyset$ is its unit. We always consider structures modulo
these equivalences. To reduce parentheses, we assume that
``$\circwhiteshort$'' and ``$\circblackshort$'' bind tighter than
``,'' which binds tighter than ``$\struct$''. Thus, we write
$\circblackshort X, Y \struct Z$ to mean $(\circblack X, Y) \struct
Z$. 

A {\em nested sequent} is a structure of the form $X \struct Y$.
This notion of nested sequents generalises Kashima's nested
sequents~\cite{kashima-cut-free-tense} for classical tense logics,
Br\"unnler's nested sequents~\cite{DBLP:conf/tableaux/BrunnlerS09}
and Poggiolesi's tree-hypersequents~\cite{poggiolesi2009} for classical modal
logics.
Figure~\ref{fig:interpretation} shows the formula-translation of
nested sequents. On both sides of the sequent, $\circ$ is interpreted
as a white (modal) operator and $\bullet$ as a black (tense) operator.
Note that however, on the lefthand side of the sequent, 
$\struct$ is interpreted as exclusion, while on the righthand side,
it is interpreted as implication.

The occurrence of a formula $A$ in a structure can have three different
{\em polorities}: {\em neutral}, {\em positive} or {\em negative}. These
are defined inductively below:
\begin{itemize}
\item The occurrence of $A$ in $X$ is neutral if it does not
occur in the scope of the structural connective $\struct.$
\item If the occurrence of $A$ in $X$ is neutral then
it is positive in $Y \struct X$ and negative in $X \struct Y$
for any structure $Y.$
\item If the occurrence of $A$ in $X$ is positive (resp. negative)
then it is also positive (resp. negative) in $X \struct Y$, $Y \struct X$,
$(X,Y)$, $(Y,X)$, $\circ X$ and $\bullet X$, for any structure $Y.$
\end{itemize}
The polarity of a structure occurrence $X$ in another structure $Y$ is defined analogously, substituting
the formula occurrence for the structure occurrence $X$. 
Note that as a consequence of the overloading of $\struct$ to represent
the structural proxies for both $\impl$ and $\dimpl$, further nesting 
of a negative context within $\struct$ does not change its polarity.

A {\em context} is a structure with a hole or a placeholder $[]$.
Contexts are ranged over by $\Sigma[]$. We write $\Sigma[X]$ for
the structure obtained by filling the hole $[]$ in the context $\Sigma[]$
with a structure $X.$ The notion of polarities of a context is
defined as above, treating the hole in the context as a formula occurrence.
We say that a context $\Sigma[~]$ is {\em neutral} if the hole $[~]$ 
has neutral polarity, {\em positive} if it has positive polarity, and {\em negative} if
it has negative polarity. 
Thus, the hole in a neutral context is never under the scope of
$\struct$.
We write $\Sigma^-[]$ to indicate that $\Sigma[]$ is 
a negative context and $\Sigma^+[]$ to indicate that
it is a positive context. Intuitively, if one views a nested
sequent as a tree (with structural connectives and multisets of formulae
as nodes), then a hole in a context is negative if it appears
to the left of the closest ancestor node labelled with $\struct.$



The context $\Sigma[]$ is {\em strict} if 
it has any  of the forms:
$$
\Sigma'[X \struct []]
\qquad
\Sigma'[[] \struct X]
\qquad
\Sigma'[\circwhiteshort{[]}]
\qquad
\Sigma'[\circblackshort{[]}]
$$ 
Intuitively, in the formation tree of a strict context, the hole
must be an immediate child of $\struct$ or
$\circwhiteshort$ or $\circblackshort$. This notion of strict contexts
will be used in later in Section~\ref{sec:calculi}.

\begin{example}
The context
$\circblack{[], (X \struct Y)}$ is a neutral context but
$\circblack{([], X) \struct Y}$ is not.
Both 
$\circblack{[], (X \struct Y)} \struct Z$ 
and
$\circblack{([], X) \struct Y} \struct Z$ 
are negative contexts. The context
$\circblackshort{[]} \struct Z$ is a strict context but
$\circblack{([], X) \struct Y} \struct Z$ is not.
\end{example}

\section{Nested Sequent Calculi}\label{sec:calculi}

\begin{figure}
{\small
{\bf Identity and logical constants:}
$$
\begin{array}{c@{\qquad\qquad}c@{\qquad\qquad}c}
\AxiomC{}
\RightLabel{$id$} \UnaryInfC{$X, A \struct A, Y$}
\DisplayProof
&
\AxiomC{}
\RightLabel{$\bot_L$} \UnaryInfC{$X, \bot \struct Y$}
\DisplayProof
&
\AxiomC{$$}
\RightLabel{$\top_R$} \UnaryInfC{$X \struct \top, Y$}
\DisplayProof
\end{array}
$$

{\bf Structural rules: }
$$
\begin{array}{c@{\qquad\qquad}c@{\qquad\qquad}c@{\qquad\qquad}c}
\AxiomC{$X \struct Z$}
\RightLabel{$w_L$} \UnaryInfC{$X, Y \struct Z$}
\DisplayProof
&
\AxiomC{$X \struct Z$}
\RightLabel{$w_R$} \UnaryInfC{$X \struct Y, Z$}
\DisplayProof
&
\AxiomC{$X, Y, Y \struct Z$}
\RightLabel{$c_L$} \UnaryInfC{$X, Y \struct Z$}
\DisplayProof
&
\AxiomC{$X \struct Y, Y, Z$}
\RightLabel{$c_R$} \UnaryInfC{$X \struct Y, Z$}
\DisplayProof
\end{array}
$$
$$
\begin{array}{c@{\qquad}c@{\qquad}c@{\qquad}c}
\AxiomC{$(X_1 \struct Y_1), X_2 \struct Y_2$}
\RightLabel{$s_L$} \UnaryInfC{$X_1, X_2 \struct Y_1, Y_2$}
\DisplayProof
&
\AxiomC{$X_1 \struct Y_1, (X_2 \struct Y_2)$}
\RightLabel{$s_R$} \UnaryInfC{$X_1, X_2 \struct Y_1, Y_2$}
\DisplayProof
&
\AxiomC{$X_2 \struct Y_2, Y_1$}
\RightLabel{$\struct_L$} \UnaryInfC{$(X_2 \struct Y_2) \struct Y_1$}
\DisplayProof
&
\AxiomC{$X_1, X_2 \struct Y_2$}
\RightLabel{$\struct_R$} \UnaryInfC{$X_1 \struct (X_2 \struct Y_2)$}
\DisplayProof
\\[2em]
\AxiomC{$\circblackshort{X} \struct Y$}
\RightLabel{$rp_\circ$} \doubleLine \UnaryInfC{$X \struct \circwhiteshort{Y}$}
\DisplayProof
&
\AxiomC{$\circwhiteshort{X} \struct Y$}
\RightLabel{$rp_\bullet$} \doubleLine \UnaryInfC{$X \struct \circblackshort{Y}$}
\DisplayProof
&
\multicolumn{2}{c}{
\AxiomC{$X_1 \struct Y_1, A$}
\AxiomC{$A, X_2 \struct Y_2$}
\RightLabel{$cut$} \BinaryInfC{$X_1, X_2 \struct Y_1, Y_2$}
\DisplayProof
}
\end{array}
$$


{\bf Logical rules:}
$$
\begin{array}{l@{\qquad\qquad}l}
\AxiomC{$X, B_i \struct Y$}
\RightLabel{$\land_L \mbox{  } i \in \{1, 2\}$} \UnaryInfC{$X, B_1 \land B_2 \struct Y$}
\DisplayProof
&
\AxiomC{$X \struct A, Y$}
\AxiomC{$X \struct B, Y$}
\RightLabel{$\land_R$} \BinaryInfC{$X \struct A \land B, Y$}
\DisplayProof
\\[2em]
\AxiomC{$X, A \struct Y$}
\AxiomC{$X, B \struct Y$}
\RightLabel{$\lor_L$} \BinaryInfC{$X, A \lor B \struct Y$}
\DisplayProof
&
\AxiomC{$X \struct B_i, Y$}
\RightLabel{$\lor_R \mbox{  } i \in \{1, 2\}$} \UnaryInfC{$X \struct B_1 \lor B_2, Y$}
\DisplayProof
\\[2em]
\AxiomC{$X \struct A, Y$}
\AxiomC{$X, B \struct Y$}
\RightLabel{$\impl_L$} \BinaryInfC{$X, A \impl B \struct Y$}
\DisplayProof
&
\AxiomC{$X, A \struct B$}
\RightLabel{$\impl_R$} \UnaryInfC{$X \struct A \impl B, Y$}
\DisplayProof
\\[2em]
\AxiomC{$A \struct B, Y$}
\RightLabel{$\dimpl_L$} \UnaryInfC{$X, A \dimpl B \struct Y$}
\DisplayProof
&
\AxiomC{$X \struct A, Y$}
\AxiomC{$X, B \struct Y$}
\RightLabel{$\dimpl_R$} \BinaryInfC{$X \struct A \dimpl B, Y$}
\DisplayProof
\end{array}
$$
\\[1em]
$$
\begin{array}{l@{\qquad\qquad}l@{\qquad\qquad}l@{\qquad\qquad}l}
\AxiomC{$A \struct X$}
\RightLabel{$\boxf_L$} \UnaryInfC{$\boxf A \struct \circwhiteshort{X}$}
\DisplayProof
&
\AxiomC{$X \struct \circwhiteshort{A}$}
\RightLabel{$\boxf_R$} \UnaryInfC{$X \struct \boxf A$}
\DisplayProof
&
\AxiomC{$A \struct X$}
\RightLabel{$\boxp_L$} \UnaryInfC{$\boxp A \struct \circblackshort{X}$}
\DisplayProof
&
\AxiomC{$X \struct \circblackshort{A}$}
\RightLabel{$\boxp_R$} \UnaryInfC{$X \struct \boxp A$}
\DisplayProof
\\[2em]
\AxiomC{$\circwhiteshort{A} \struct X$}
\RightLabel{$\diaf_L$} \UnaryInfC{$\diaf A \struct X$}
\DisplayProof
&
\AxiomC{$X \struct A$}
\RightLabel{$\diaf_R$} \UnaryInfC{$\circwhiteshort{X} \struct \diaf A$}
\DisplayProof
&
\AxiomC{$\circblackshort{A} \struct X$}
\RightLabel{$\diap_L$} \UnaryInfC{$\diap A \struct X$}
\DisplayProof
&
\AxiomC{$X \struct A$}
\RightLabel{$\diap_R$} \UnaryInfC{$\circblackshort{X} \struct \diap A$}
\DisplayProof
\end{array}
$$
}
\caption{$\lbikt$: a shallow inference system for $\bikt$}
\label{fig:lbikt}
\end{figure}

\begin{figure}
{\small
{\bf Identity and logical constants:}

$$
\begin{array}{c@{\qquad\qquad}c@{\qquad\qquad}c}
\AxiomC{}
\RightLabel{$id$} \UnaryInfC{$\Sigma[X, A \struct A, Y]$}
\DisplayProof
&
\AxiomC{}
\RightLabel{$\bot_L$} \UnaryInfC{$\Sigma[\bot, X \struct Y]$}
\DisplayProof
&
\AxiomC{}
\RightLabel{$\top_R$} \UnaryInfC{$\Sigma[X \struct \top, Y]$}
\DisplayProof
\end{array}
$$

{\bf Propagation rules:}

$$
\begin{array}{c@{\qquad}c}
\AxiomC{$\Sigma^-[A, (A, X \struct Y)]$}
\RightLabel{$\struct_{L1}$} \UnaryInfC{$\Sigma^-[A, X \struct Y]$}
\DisplayProof
&
\AxiomC{$\Sigma^+[(X \struct Y, A), A]$}
\RightLabel{$\struct_{R1}$} \UnaryInfC{$\Sigma^+[X \struct Y,A]$}
\DisplayProof
\\[2em]
\AxiomC{$\Sigma[X, A \struct W, (A, Y \struct Z)]$}
\RightLabel{$\struct_{L2}$} \UnaryInfC{$\Sigma[X, A \struct W, (Y \struct Z)]$}
\DisplayProof
&
\AxiomC{$\Sigma[(X \struct Y, A), W \struct A, Z]$}
\RightLabel{$\struct_{R2}$} \UnaryInfC{$\Sigma[(X \struct Y), W \struct A, Z]$}
\DisplayProof
\\[2em]
\AxiomC{$\Sigma^-[A, \circblack{\boxf A,X}]$}
\RightLabel{$\boxf_{L1}$} \UnaryInfC{$\Sigma^-[\circblack{\boxf A,X}]$}
\DisplayProof
\qquad
\AxiomC{$\Sigma^+[A, \circblack{\diaf A, X}]$}
\RightLabel{$\diaf_{R1}$} \UnaryInfC{$\Sigma^+[\circblack{\diaf A, X}]$}
\DisplayProof
&
\AxiomC{$\Sigma^-[A, \circwhite{\boxp A,X}]$}
\RightLabel{$\boxp_{L1}$} \UnaryInfC{$\Sigma^-[\circwhite{\boxp A,X}]$}
\DisplayProof
\qquad
\AxiomC{$\Sigma^+[A, \circwhite{\diap A,X}]$}
\RightLabel{$\diap_{R1}$} \UnaryInfC{$\Sigma^+[\circwhite{\diap A,X}]$}
\DisplayProof
\\[2em]
\AxiomC{$\Sigma[\boxp A,X \struct \circblack{A \struct Y},Z]$}
\RightLabel{$\boxp_{L2}$} \UnaryInfC{$\Sigma[\boxp A,X \struct \circblackshort{Y}, Z]$}
\DisplayProof
&
\AxiomC{$\Sigma[\circwhite{X \struct A}, Y \struct Z, \diaf A]$}
\RightLabel{$\diaf_{R2}$} \UnaryInfC{$\Sigma[\circwhiteshort{X}, Y \struct Z, \diaf A]$}
\DisplayProof
\\[2em]
\AxiomC{$\Sigma[\boxf A, X \struct \circwhite{A \struct Y}, Z]$}
\RightLabel{$\boxf_{L2}$} \UnaryInfC{$\Sigma[\boxf A, X \struct \circwhiteshort{Y}, Z]$}
\DisplayProof
&
\AxiomC{$\Sigma[\circblack{X \struct A}, Y \struct Z, \diap A]$}
\RightLabel{$\diap_{R2}$} \UnaryInfC{$\Sigma[\circblackshort{X}, Y \struct Z, \diap A]$}
\DisplayProof
\end{array}
$$


{\bf Logical rules:}

$$
\begin{array}{c@{\qquad\qquad\qquad}c}
\AxiomC{$\Sigma^-[A \lor B, A]$}
\AxiomC{$\Sigma^-[A \lor B, B]$}
\RightLabel{$\lor_L$} \BinaryInfC{$\Sigma^-[A \lor B]$}
\DisplayProof
&
\AxiomC{$\Sigma^+[A \lor B, A, B]$}
\RightLabel{$\lor_R$} \UnaryInfC{$\Sigma^+[A \lor B]$}
\DisplayProof
\\[2em]
\AxiomC{$\Sigma^-[A \land B, A, B]$}
\RightLabel{$\land_L$} \UnaryInfC{$\Sigma^-[A \land B]$}
\DisplayProof
&
\AxiomC{$\Sigma^+[A \land B, A]$}
\AxiomC{$\Sigma^+[A \land B, B]$}
\RightLabel{$\land_R$} \BinaryInfC{$\Sigma^+[A \land B]$}
\DisplayProof
\\[2em]
\AxiomC{$\Sigma^-[A \dimpl B, (A \struct B)]$}
\RightLabel{$\dimpl_L$} \UnaryInfC{$\Sigma^-[A \dimpl B]$}
\DisplayProof
&
\AxiomC{$\Sigma^+[A \impl B, (A \struct B)]$}
\RightLabel{$\impl_R$} \UnaryInfC{$\Sigma^+[A \impl B]$}
\DisplayProof
\\[2em]
\multicolumn{2}{c}{
\AxiomC{$\Sigma^-[X, A \impl B \struct A]$}
\AxiomC{$\Sigma^-[X, A \impl B, B]$}
\RightLabel{$\impl_L$} \BinaryInfC{$\Sigma^-[X, A \impl B]$}
\DisplayProof
\qquad 
\hbox{$\Sigma^-[]$ is a strict context}
}
\\[2em]
\multicolumn{2}{c}{
\AxiomC{$\Sigma^+[X, A \dimpl B, A]$}
\AxiomC{$\Sigma^+[B \struct X, A \dimpl B]$}
\RightLabel{$\dimpl_R$} \BinaryInfC{$\Sigma^+[X, A \dimpl B]$}
\DisplayProof
\qquad
\hbox{$\Sigma^+[]$ is a strict context}
}
\\[2em]
\AxiomC{$\Sigma^-[\diaf A, \circwhiteshort{A}]$}
\RightLabel{$\diaf_L$} \UnaryInfC{$\Sigma^-[\diaf A]$}
\DisplayProof
\qquad
\AxiomC{$\Sigma^+[\boxf A, \circwhiteshort{A}]$}
\RightLabel{$\boxf_R$} \UnaryInfC{$\Sigma^+[\boxf A]$}
\DisplayProof
&
\AxiomC{$\Sigma^-[\diap A, \circblackshort{A}]$}
\RightLabel{$\diap_L$} \UnaryInfC{$\Sigma^-[\diap A]$}
\DisplayProof
\qquad
\AxiomC{$\Sigma^+[\boxp A, \circblackshort{A}]$}
\RightLabel{$\boxp_R$} \UnaryInfC{$\Sigma^+[\boxp A]$}
\DisplayProof
\end{array}
$$
}
\caption{$\dbikt$: a deep inference system for $\bikt$}
\label{fig:dbikt}
\end{figure}

We now present the two nested sequent calculi that we will use in the
rest of the paper: a shallow inference calculus $\lbikt$ and a deep
inference calculus $\dbikt$.

Fig.~\ref{fig:lbikt} gives the rules of the shallow inference calculus
$\lbikt$. Central to this calculus is the idea that inference rules
can only be applied to formulae at the top level of nested sequents,
and the structural rules $s_L$, $s_R$, $\struct_L$, $\struct_R$,
$rp_\circ$ and $rp_\bullet$, also called the {\em residuation rules},
are used to bring the required sub-structures to the top level. 
These rules are similar to residuation postulates in display logic, 
are essential for the cut-elimination proof of $\lbikt$, but  contain too much
non-determinism for effective proof search. 
Another issue with proof search in $\lbikt$ is the structural contraction rules,
which allow contraction on arbitrary structures, not just formulae as
in traditional sequent calculi. 

$\lbikt$ is as a merger of two calculi: the LBiInt
calculus~\cite{gorepostniecetiu2008,postniece2009} for the
intuitionistic connectives, and the display
calculus~\cite{Gore98IGPL} 
for the tense connectives.

Note that we use $\circ$ and $\bullet$ as structural
proxies for the non-residuated pairs $(\diaf, \boxf)$ and $(\diap,
\boxp)$ respectively, whereas Wansing~\cite{Wansing1994} uses only one
$\bullet$ as a structural proxy for the residuated pair $(\diap,
\boxf)$ and recovers $(\diaf, \boxp)$ via classical negation,
while Gor\'e~\cite{Gore98IGPL} uses $\circ$ and $\bullet$ as structural proxies for the
residuated pairs $(\diaf,\boxp)$ and $(\diap,\boxf)$ respectively.
As we shall see later, our choice allows
us to retain the modal fragment $(\diaf,\boxf)$ 
by simply eliding all rules that contain ``black'' operators from
our deep sequent calculus.

Fig.~\ref{fig:dbikt} gives the rules of the deep inference calculus
$\dbikt$. Here the inference rules can be applied at any level of the
nested sequent, indicated by the use of contexts. Notably, there are
no residuation rules; indeed one of the goals of our paper is to show
that the residuation rules of $\lbikt$ can be simulated by deep
inference and propagation rules in $\dbikt$. Another feature of
$\dbikt$ is the use of polarities in defining contexts to which rules
are applicable. For example, the premise of the $\boxf_{L1}$ rule
denotes a negative context $\Sigma$ which itself contains a formula
$A$ and a $\circblackshort$-structure, such that the
$\circblackshort$-structure contains $\boxf A$.

$\dbikt$ achieves the goal of merging the DBiInt
calculus~\cite{postniece2009} and a two-sided version of the DKt
calculus~\cite{gorepostniecetiu2009}. While in the shallow inference
case, a calculus for \bikt{} could be obtained relatively easily by
merging shallow inference calculi for BiInt and tense logics, the
combination of calculi is not so obvious in the deep inference
case. Although the propagation rules for $\struct$-structures remain
the same as in the BiInt case~\cite{postniece2009}, the propagation
rules for $\circ$- and $\bullet$-structures are not as simple as in
the DKt calculus~\cite{gorepostniecetiu2009}. 
Since we do not assume any direct relationship between $\boxf$ and
$\diaf$, or $\boxp$ and $\diap$, propagation rules 
like $\boxp_{L2}$ need to involve the $\struct$ structural connective so
they can refer to both sides of the nested sequent.

Note that in the rules $\impl_L$ and $\dimpl_R$ in $\dbikt$,
we require that the contexts in which the principal formulae
reside are strict contexts. This is strictly speaking not necessary, i.e.,
we could remove the proviso without affecting the expressivity of the proof
system. The proviso does,
however, reduce the non-determinism in partitioning the 
contexts in $\impl_L$ or $\dimpl_R.$
Consider, for example, the nested sequent
$\circ(a, b \impl c) \struct d.$ Without the requirement of strict contexts,
there are two instances of $\impl_L$ with that nested sequent as the conclusion:
$$
\infer[\impl_L]
{\circ(a, b \impl c) \struct d}
{\circ(a, (b \impl c \struct b)) \struct d
 &
 \circ(a, b \impl c, c) \struct d
}
$$
$$
\infer[\impl_L]
{\circ(a, b \impl c) \struct d}
{
 \circ(a, b \impl c \struct b) \struct d
 &
 \circ(a, b \struct c, c) \struct d
}
$$
In the first instance, the context is $\circ(a, [~]) \struct d$, which is
not strict, whereas in the second instance, it is $\circ([~])\struct d$,
which is strict.  
In general, if there are $n$ formulae connected to $b\impl c$ via the disjunctive structural
connective, then there are $2^n$ possible instances of $\impl_L$
without the strict context proviso. 

We write $\deriveshallow{\pi}{X \struct Y}$ when $\pi$ is a 
derivation of the shallow sequent $X \struct Y$ in $\lbikt$, and
$\derivedeep{\pi}{X \struct Y}$ when $\pi$ is a derivation of the
sequent $X \struct Y$ in $\dbikt$. In either calculus,
the height $|\pi|$ of
a derivation $\pi$ is the number of sequents on the longest branch.

\begin{example}
  Below we derive Ewald's axiom 9 for $IK_t$~\cite{ewald1986} in
  $\lbikt$ and $\dbikt$. The $\lbikt$-derivation on the left read
  bottom-up brings the required sub-structure $\diap A$ to the
  top-level using the residuation rule $rp_\circwhiteshort$ and
  applies $\diap_R$ backward. The $\dbikt$-derivation on the right instead applies $\boxf_R$
  deeply, and propagates the required formulae to the appropriate
  sub-structure using $\diap_{R1}$. Note that contraction is implicit
  in $\diap_{R1}$, and all propagation rules.
$$
\begin{array}{c@{\qquad\qquad}c}
\AxiomC{$$}
\RightLabel{$id$} \UnaryInfC{$A \struct A$}
\RightLabel{$\diap_R$} \UnaryInfC{$\circblackshort{A} \struct \diap A$}
\RightLabel{$rp_\circ$} \UnaryInfC{$A \struct \circwhiteshort{\diap A}$}
\RightLabel{$\boxf_R$} \UnaryInfC{$A \struct \boxf \diap A$}
\RightLabel{$\impl_R$} \UnaryInfC{$\emptyset \struct A \rightarrow \boxf \diap A$}
\DisplayProof
&
\AxiomC{$$}
\RightLabel{$id$} \UnaryInfC{$\emptyset \struct (A \struct A, \circwhite{\diap A})$}
\RightLabel{$\diap_{R1}$} \UnaryInfC{$\emptyset \struct (A \struct \circwhite{\diap A})$}
\RightLabel{$\boxf_R$} \UnaryInfC{$\emptyset \struct (A \struct \boxf \diap A)$}
\RightLabel{$\impl_R$} \UnaryInfC{$\emptyset \struct A \rightarrow \boxf \diap A$}
\DisplayProof
\end{array}
$$
\end{example}

\paragraph{Display property}

A (deep or shallow) nested sequent can be seen as a tree of traditional
sequents.
The structural rules of $\lbikt$ allows shuffling of structures
to display/un-display a particular node in the tree, so
inference rules can be applied to it. This is similar to
the display property in traditional display calculi, where
any substructure can be displayed and un-displayed. 
We state the display property of $\lbikt$ more precisely
in subsequent lemmas. We shall use two 
``display'' rules which are easily derivable using
$s_L$, $s_R$, $\struct_L$ and $\struct_R$:
$$
\AxiomC{$(X_1 \struct X_2) \struct Y$}
\RightLabel{$rp_\struct^L$} \doubleLine \UnaryInfC{$X_1 \struct X_2, Y$}
\DisplayProof
\qquad
\AxiomC{$X_1 \struct (X_2 \struct Y)$}
\RightLabel{$rp_\struct^R$} \doubleLine \UnaryInfC{$X_1, X_2 \struct Y$}
\DisplayProof
$$
Let $DP = \{rp_\struct^L, rp_\struct^R, rp_\circ, rp_\bullet\}$ and
let $DP$-derivable mean ``derivable using rules only from $DP$''.

The following lemmas can be proved by simple induction on the
size of the context $\Sigma[~].$

\begin{lemma}[Display property for neutral contexts]\label{lemma:display-property-neutral}
Let $\Sigma[]$ be a neutral context. 
Let $X$ be a structure and $p$ a propositional variable 
not occurring in $X$ nor $\Sigma[].$
Then there exist structures $Y$ and $Z$ such that:
\begin{enumerate}
\item $Y \struct p$ is DP-derivable from 
$X \struct \Sigma[p]$ 
and 

\item $p \struct Z$ is 
DP-derivable from $\Sigma[p] \struct X.$ 
\end{enumerate}
\end{lemma}

\begin{lemma}[Display property for positive contexts]\label{lemma:display-property-positive}
Let $\Sigma[]$ be a positive context. Let $X$ be
a structure and $p$ a propositional variable not 
occurring in $X$ nor $\Sigma[].$
Then there exist structures $Y$ and $Z$ such that:
\begin{enumerate}
\item $Y \struct p$ is $DP$-derivable from 
$X \struct \Sigma[p]$, and 

\item $Z \struct p$ is 
$DP$-derivable from $\Sigma[p] \struct X.$ 
\end{enumerate}
\end{lemma}

\begin{lemma}[Display property for negative contexts]\label{lemma:display-property-negative}
Let $\Sigma[]$ be a negative context. Let $X$ be
a structure and $p$ a propositional variable 
not occurring in  $X$ nor $\Sigma[].$
Then there exist structures $Y$ and $Z$ such that:
\begin{enumerate}
\item $p \struct Y$ is $DP$-derivable from 
$X \struct \Sigma[p]$ 
and 

\item $p \struct Z$ is 
$DP$-derivable from $\Sigma[p] \struct X.$ 
\end{enumerate}
\end{lemma}

Since the rules in $DP$ are all invertible,
the derivations constructed in the above lemmas 
are invertible derivations. That is, we can derive
$Y \struct p$ from $X \struct \Sigma[p]$
and vice versa. Note also that since rules in the
shallow system are closed under substitution, this also means
$Y \struct Z$ is derivable from $X \struct \Sigma[Z]$,
and vice versa, for any $Z.$

The display property of pure display calculi is the
ability to display/un-display a structure with respect to a
top-level turnstile $\vdash$ (say) as the {\em whole} of the
antecedent or succedent.  For example, we have to display $V \struct
W$ as the whole of the antecedent or succedent as $V \struct W \vdash
Z$ or $Z \vdash V \struct W$.  Our shallow nested sequent calculus 
instead enables us to ``zoom in'' to $V \struct W$ in $X \struct Y$ by
explicitly transforming the latter into $X', V \struct W, Y'$
so a rule can be applied to any top-level formula/structure of $V$ or
$W$. 
Our deep nested sequent calculus 
allows us to ``zoom in'' to $V \struct W$ by treating it as
the filler of a hole $\Sigma[V \struct W]$, without explicit
transformations.

\section{Cut elimination in $\lbikt$}
\label{sec:cut-elim}

\begin{figure}[t]
{\small
$$
\begin{array}{cccc}
\AxiomC{$\psi_1$}
\noLine \UnaryInfC{$X_1' \struct A$}
\RightLabel{$\diaf_R$} \UnaryInfC{$\circwhite{X_1'} \struct \diaf A$}
\noLine \UnaryInfC{$\vdots$}
\noLine \UnaryInfC{$X_1 \struct Y_1, \diaf A$}
\DisplayProof
\quad 
&
\quad
\AxiomC{$\psi_2$}
\noLine \UnaryInfC{$\circwhiteshort A \struct Y_2'$}
\RightLabel{$\diaf_L$} \UnaryInfC{$\diaf A \struct Y_2'$}
\noLine \UnaryInfC{$\vdots$}
\noLine \UnaryInfC{$\diaf A, X_2 \struct Y_2$}
\DisplayProof &
\quad
\AxiomC{$\circwhiteshort{X_1'} \struct (X_2 \struct Y_2)$}
\noLine \UnaryInfC{$\vdots$}
\noLine \UnaryInfC{$X_1 \struct Y_1, (X_2 \struct Y_2)$}
\RightLabel{$s_R$} \UnaryInfC{$X_1, X_2 \struct Y_1, Y_2$}
\DisplayProof
\quad
& 
\quad
\AxiomC{$\psi_1$}
\noLine \UnaryInfC{$X_1' \struct A$}
\AxiomC{$\psi_2$}
\noLine \UnaryInfC{$\circwhiteshort A \struct Y_2'$}
\RightLabel{$rp_\bullet$} \UnaryInfC{$A \struct \circblackshort{Y_2'}$}
\RightLabel{$cut$} \BinaryInfC{$X_1' \struct \circblackshort{Y_2'}$}
\RightLabel{$rp_\bullet$} \UnaryInfC{$\circwhiteshort{X_1'} \struct Y_2'$}
\noLine \UnaryInfC{$\vdots$}
\noLine \UnaryInfC{$\circwhiteshort{X_1'}, X_2 \struct Y_2$}
\RightLabel{$\struct_R$} \UnaryInfC{$\circwhiteshort{X_1'} \struct (X_2 \struct Y_2)$}
\DisplayProof \\ \\
(1) & (2) & (3) & (4)
\end{array} 
$$
}
\caption{An example of cut reduction}
\label{fig:cut}
\end{figure}

Our cut-elimination proof is based on the method of proof-substitution
presented in~\cite{gorepostniecetiu2008}. It is very similar
to the general cut elimination method used in display calculi.
The proof relies on the display property and the fact that
inference rules in $\lbikt$ are closed under substitutions.

We illustrate the cut reduction steps here with an example.
Consider the derivation below ending with a cut on $\diaf A$:
$$
\AxiomC{$\pi_1$}
\noLine \UnaryInfC{$X_1 \struct Y_1, \diaf A$}
\AxiomC{$\pi_2$}
\noLine \UnaryInfC{$\diaf A, X_2 \struct Y_2$}
\RightLabel{$cut$} \BinaryInfC{$X_1, X_2 \struct Y_1, Y_2$}
\DisplayProof
$$
Instead of permuting the cut rule locally, we trace the cut formula
$\diaf A$ until it becomes principal in the derivations $\pi_1$ and
$\pi_2$, and then apply cut on a smaller formula. Suppose
that $\pi_1$ and $\pi_2$ are respectively the derivations (1) and (2)
in Figure~\ref{fig:cut}. 
We first transform $\pi_1$ by substituting $(X_2 \struct Y_2)$ for
$\diaf A$ in $\pi_1$ and obtain the sub-derivation with an
open leaf as shown in Figure~\ref{fig:cut}(3). 
We then prove the open leaf by uniformly substituting
$\circwhite{X_1'}$ for $\diaf A$ in $\pi_2$, and applying cut on a
sub-formula $A$, as shown in Figure~\ref{fig:cut}(4).

The {\em cut rank} of an instance of cut 
is the size of the cut formula, as usual.
The cut rank $cr(\pi)$ of a derivation $\pi$
is the largest cut rank of the cut instances
in $\pi$ (or zero, if $\pi$ is cut-free). 
Given a formula $A$, we denote with $|A|$ its size.

To formalise the cut elimination proof, we first
introduce a notion of multiple-hole contexts. 
A {\em $k$-hole context} is a context with $k$ holes.
Given a $k$-hole context $\Sigma[\cdots]$
we write $\Sigma[X^{k}]$ to stand for the structure obtained from
$\Sigma[\cdots]$ by replacing each hole with an occurrence of the
structure $X$. For example, if $\Sigma[] = ([], \circblack{[], W
  \struct Y}) \struct Z$ then $\Sigma[(\circwhiteshort{U})^2] =
(\circwhiteshort{U}, \circblack{\circwhiteshort{U}, W \struct Y})
\struct Z$.
A {\em neutral} (resp. {\em positive} and {\em negative}) $k$-hole context is
a $k$-hole context where all the holes have neutral (resp. positive and negative) 
polarity. 
A {\em quasi-positive} (resp. {\em quasi-negative}) $k$-hole
context is a $k$-hole context where each hole in the context
has either neutral or positive (resp. negative) polarity. 
Obviously, a $k$-hole positive (negative) context 
is also a $k$-hole quasi-positive (quasi-negative) context.

Lemma~\ref{lm:cut-atm} states the proof substitutions needed to
eliminate atomic cuts. Lemmas~\ref{lm:cut-or}-\ref{lm:cut-diap} state
the proof substitutions needed for non-atomic cuts. We only give the
proofs of the cases involving the modal connectives as the other
proofs are unchanged from~\cite{gorepostniecetiu2008}.

\begin{lemma}
\label{lm:cut-atm}
Suppose $p, X \struct Y$ is cut-free derivable for some fixed $p$, $X$ and
$Y$. Then for any $k$-hole positive context $Z_1[\cdots]$ and any
$l$-hole quasi-positive context $Z_2[\cdots]$, if
$Z_1[p^k] \struct Z_2[p^l]$ is cut-free derivable, then $Z_1[(X \struct Y)^k] \struct Z_2[(X \struct Y)^l]$
is cut-free derivable.
\end{lemma}
\begin{proof}
 Let $\pi_1$ be a cut-free derivation of $p,X \struct Y$ and let $\pi_2$
 be a cut-free derivation of $Z_1[p^k] \struct Z_2[p^l].$ We construct
 a cut-free derivation $\pi$ of $Z_1[(X \struct Y)^k] \struct Z_2[(X \struct Y)^l]$ 
 by induction on $|\pi_2|.$ Most cases follow
 straightforwardly from the induction hypothesis. The only
 non-trivial case is when $p$ is active in the derivation, i.e., when
 $\pi_2$ ends with an $id$ rule or a contraction rule applied to an
 occurence of $p$ to be substituted for:
\begin{itemize}
\item Suppose $\pi_2$ is 
$$
\AxiomC{$$}
\RightLabel{$id$} \UnaryInfC{$Z_1'[p^k], p \struct p, Z_2'[p^{l-1}]$}
\DisplayProof
$$
Note that the $p$ immediately to the left of the ``$\struct$'' cannot be
part of the $p^k$ by the restrictions on the context $Z_1[\cdots]$. The
derivation $\pi$ is then constructed as follows, where we use double lines to abbreviate derivations:
$$
\AxiomC{$\pi$}
\noLine \UnaryInfC{$p, X \struct Y$}
\RightLabel{$\struct_R$} \UnaryInfC{$p \struct (X \struct Y)$}
\RightLabel{$w_R;w_L$} \UnaryInfC{$Z_1'[(X \struct Y)^k], p \struct (X \struct Y), Z_2'[(X \struct Y)^{l-1}]$}
\DisplayProof
$$

\item Suppose $\pi_2$ is
$$
\AxiomC{$\psi$}
\noLine \UnaryInfC{$Z_1[p^k] \vdash p, p, Z_2'[p^{l-1}]$}
\RightLabel{$c_R$} \UnaryInfC{$Z_1[p^k] \vdash p, Z_2'[p^{l-1}]$}
\DisplayProof
$$
By induction hypothesis, we have a cut-free derivation $\psi'$ of
$$Z_1[(X \struct Y)^k] \struct (X \struct Y), (X \struct Y), Z_2'[(X \struct Y)^{l-1}].$$
The derivation $\pi$ is then constructed as follows:
$$
\AxiomC{$\psi'$}
\noLine \UnaryInfC{$Z_1[(X \struct Y)^k] \struct (X \struct Y), (X \struct Y), Z_2'[(X \struct Y)^{l-1}]$}
\RightLabel{$c_R$} \UnaryInfC{$Z_1[(X \struct Y)^k] \struct (X \struct Y), Z_2'[(X \struct Y)^{l-1}]$}
\DisplayProof
$$
\end{itemize}
\end{proof}

\begin{lemma}
\label{lm:cut-or}
Suppose $\deriveshallow{\pi_i}{X \struct Y, A_i}$, for some 
$i \in \{1,2\}$, such that $cr(\pi_i) < |A_1 \lor A_2|$. Suppose
$\deriveshallow{\pi_3}{Z_1[(A_1 \lor A_2)^k] \struct Z_2[(A_1 \lor A_2)^l]}$ 
for some $k$-hole quasi-negative context $Z_1[\cdots]$ and
$l$-hole negative context $Z_2[\cdots]$, such that $cr(\pi_3) < |A_1
\lor A_2|.$ Then there exists $\pi$ such that
$\deriveshallow{\pi}{Z_1[(X \struct Y)^k] \struct Z_2[(X \struct Y)^l]}$ 
and $cr(\pi) < |A \lor B|$.
\end{lemma}
\begin{proof}
By induction on $|\pi_3|.$ In the following, we let $A = A_1 \lor A_2.$
Most cases follow straightforwardly from the induction hypothesis. The only interesting case is when a left-rule is applied to an occurrence of $A_1 \lor A_2$ which is to be replaced by $X \struct Y.$ 
That is, $\pi_3$ is 
$$
\AxiomC{$\psi_1$}
\noLine \UnaryInfC{$Z_1'[A^{k-1}], A_1 \struct Z_2[A^l]$}
\AxiomC{$\psi_2$}
\noLine \UnaryInfC{$Z_1'[A^{k-1}], A_2 \struct Z_2[A^l]$}
\RightLabel{$\lor_L$} \BinaryInfC{$Z_1'[A^{k-1}], A_1 \lor A_2 \struct Z_2[A^l]$}
\DisplayProof
$$
By induction hypothesis, we have a derivation $\psi_i'$, for each $i \in \{1,2\}$, of 
$$
Z_1'[(X \struct Y)^{k-1}], A_i \struct Z_2[(X \struct Y)^l] 
$$
with $cr(\psi_i') < |A_1 \lor A_2|.$ The derivation $\pi$ is then
constructed as follows:
$$
\AxiomC{$\pi_i$}
\noLine \UnaryInfC{$X \struct Y, A_i$}
\RightLabel{$\struct_L$} \UnaryInfC{$X \struct Y \struct A_i$}
\AxiomC{$\psi_i'$}
\noLine \UnaryInfC{$Z_1'[(X \struct Y)^{k-1}], A_i \struct Z_2[(X \struct Y)^l]$}
\RightLabel{$cut$} \BinaryInfC{$Z_1'[(X \struct Y)^{k-1}], X \struct Y \struct Z_2[(X \struct Y)^l]$}
\DisplayProof
$$
\end{proof}

\begin{lemma}
\label{lm:cut-and}
Suppose $\deriveshallow{\pi_1}{X \struct Y, A_1}$ and
$\deriveshallow{\pi_2}{X \struct Y, A_2}$ with $cr(\pi_1) < |A_1 \land A_2|$ 
and $cr(\pi_2) < |A_1 \land A_2|.$ Suppose 
$\deriveshallow{\pi_3}{Z_1[(A_1 \land A_2)^k] \struct Z_2[(A_1 \land A_2)^l]}$ 
for some $k$-hole quasi-negative context $Z_1[\cdots]$ and
$l$-hole negative context $Z_2[\cdots]$ with 
$cr(\pi_3) < |A_1 \land A_2|.$ 
Then there exists $\pi$ such that 
$\deriveshallow{\pi}{Z_1[(X \struct Y)^k] \struct Z_2[(X \struct Y)^l]}$ and 
$cr(\pi) < |A \land B|$.
\end{lemma}
\begin{proof}
By induction on $|\pi_3|.$ In the following, we let $A = A_1 \land A_2.$ Most cases follow straightforwardly from the induction hypothesis. The only interesting case is when a left-rule is applied to an occurrence of $A_1 \land A_2$ which is to be replaced by $X \struct Y.$ That is, $\pi_3$ is:
$$
\AxiomC{$\psi_i$}
\noLine \UnaryInfC{$Z_1[A^{k-1}], A_i \struct Z_2[A^l]$}
\RightLabel{$\land_L$} \UnaryInfC{$Z_1[A^{k-1}], A_1 \land A_2 \struct Z_2[A^l]$}
\DisplayProof
$$
for some $i \in \{1,2\}$.
By induction hypothesis, we have a derivation $\psi_i'$, for some $i \in \{1,2\}$, of 
$$
Z_1[(X \struct Y)^{k-1}], A_i \struct Z_2[(X \struct Y)^l]
$$
with $cr(\psi_i') < |A_1 \land A_2|.$ The derivation $\pi$ is then
constructed as follows:
$$
\AxiomC{$\pi_i$}
\noLine \UnaryInfC{$X \struct Y, A_i$}
\RightLabel{$\struct_L$} \UnaryInfC{$X \struct Y \struct A_i$}
\AxiomC{$\psi_i'$}
\noLine \UnaryInfC{$Z_1'[(X \struct Y)^{k-1}], A_i \struct Z_2[(X \struct Y)^l$}
\RightLabel{$cut$} \BinaryInfC{$Z_1'[(X \struct Y)^{k-1}], X \struct Y \struct Z_2[(X \struct Y)^{l-1}]$}
\DisplayProof
$$
\end{proof}

\begin{lemma}
\label{lm:cut-impl}
Suppose $\deriveshallow{\pi_1}{X, A \struct B}$ and
$\deriveshallow{\pi_2}{Z_1[(A \impl B)^k] \struct Z_2[(A \impl B)^l]}$
for some $k$-hole quasi-negative context $Z_1[\cdots]$ and $l$-hole negative
context $Z_2[\cdots]$, and the cut ranks of $\pi_1$ and $\pi_2$ are
smaller than $|A \impl B|.$ Then there exists $\pi$ such that
$\deriveshallow{\pi}{Z_1[X^k] \struct Z_2[X^l]}$ and $cr(\pi) < |A \impl B|$.
\end{lemma}
\begin{proof}
By induction on $|\pi_2|$. The non-trivial case is when $\pi_2$ ends with $\impl_L$ on $A \impl B$:
$$
\AxiomC{$\psi_1$}
\noLine \UnaryInfC{$Z_1'[(A \impl B)^{k-1}] \struct A, Z_2[(A \impl B)^l]$}
\AxiomC{$\psi_2$}
\noLine \UnaryInfC{$Z_1'[(A \impl B)^{k-1}], B \struct Z_2[(A \impl B)^l]$}
\RightLabel{$\impl_L$} \BinaryInfC{$Z_1'[(A \impl B)^{k-1}], A \impl B \struct Z_2[(A \impl B)^l]$}
\DisplayProof
$$
By induction hypothesis, we have derivations $\psi_1'$ and $\psi_2'$
respectively of the sequents below where $cr(\psi_1') < |A \impl B|$ and $cr(\psi_2') < |A \impl B|$:
$$
Z_1'[X^{k-1}] \struct A, Z_2[X^l]
\qquad\qquad
Z_1'[X^{k-1}], B \struct Z_2[X^l]
$$
In
the following, we let $V_1$ denote $Z_1'[X^{k-1}]$
and $V_2$ denote $Z_2[X^l].$ The derivation $\pi$ is
constructed as follows:
$$
\AxiomC{$\psi_1'$}
\noLine \UnaryInfC{$V_1 \vdash A, V_2$}
\AxiomC{$\pi_1$}
\noLine \UnaryInfC{$X,A \struct B$}
\AxiomC{$\psi_2'$}
\noLine \UnaryInfC{$V_1, B \struct V_2$}
\RightLabel{$cut$} \BinaryInfC{$V_1,A,X \struct V_2$}
\RightLabel{$cut$} \BinaryInfC{$V_1, V_1, X \struct V_2, V_2$}
\RightLabel{$c_L;c_R$} \UnaryInfC{$V_1, X \struct V_2$}
\DisplayProof
$$
\end{proof}

\begin{lemma}
\label{lm:cut-excl}
Suppose $\deriveshallow{\pi_1}{X \struct Y, A}$ and
$\deriveshallow{\pi_2}{X, B \struct Y}$, and the cut ranks of $\pi_1$
and $\pi_2$ are smaller than $|A \dimpl B|.$ Suppose
$\deriveshallow{\pi_3}{Z_1[(A \dimpl B)^k] \struct Z_2[(A \dimpl B)^l]}$ 
for some $k$-hole quasi-negative context $Z_1[\cdots]$ and
$l$-hole negative context $Z_2[\cdots]$ with $cr(\pi_3) < |A \dimpl B|.$ 
Then there exists $\pi$ such that 
$\deriveshallow{\pi}{Z_1[(X \struct Y)^k] \struct Z_2[(X \struct Y)^l]}$ 
and $cr(\pi) < |A \dimpl B|$.
\end{lemma}
By induction on $|\pi_2|$. The non-trivial case is when $\pi_2$ ends with $\dimpl_L$ on $A \dimpl B$:

$$
\AxiomC{$\psi$}
\noLine \UnaryInfC{$A \struct B, Z_2[(A \dimpl B)^l]$}
\RightLabel{$\dimpl_L$} \UnaryInfC{$Z_1'[(A \dimpl B)^{k-1}], A \dimpl B \struct Z_2[(A \dimpl B)^l]$}
\DisplayProof
$$

By induction hypothesis, we have a derivation $\psi'$ of 
$$
A \struct B, Z_2[(X \struct Y)^l]
$$
with $cr(\psi) < |A \dimpl B|.$ The derivation $\pi$ is then
constructed as follows:

$$
\AxiomC{$\pi_1$}
\noLine \UnaryInfC{$X \struct Y, A$}
\AxiomC{$\psi'$}
\noLine \UnaryInfC{$A \struct B, Z_2[(X \struct Y)^l]$}
\AxiomC{$\pi_2$}
\noLine \UnaryInfC{$X, B \struct Y$}
\RightLabel{$cut$} \BinaryInfC{$A, X \struct Y, Z_2[(X \struct Y)^l]$}
\RightLabel{$cut$} \BinaryInfC{$X, X \struct Y, Y, Z_2[(X \struct Y)^l]$}
\RightLabel{$c_L;c_R$} \UnaryInfC{$X \struct Y, Z_2[(X \struct Y)^l]$}
\RightLabel{$\struct_L$} \UnaryInfC{$Z_1[(X \struct Y)^{k-1}], X \struct Y \struct Z_2[(X \struct Y)^l]$}
\DisplayProof
$$

\begin{lemma}
\label{lm:cut-boxf}
Suppose $\deriveshallow{\pi_1}{X \struct \circwhiteshort{A}}$ and
$\deriveshallow{\pi_2}{Z_1[(\boxf A)^k] \struct Z_2[(\boxf A)^l]}$ for
some $k$-hole quasi-negative context $Z_1[\cdots]$ and $l$-hole negative
context $Z_2[\cdots]$, and the cut ranks of $\pi_1$ and $\pi_2$ are
smaller than $|\boxf A|.$ Then there exists $\pi$ such that
$\deriveshallow{\pi}{Z_1[X^k] \struct Z_2[X^l]}$ and $cr(\pi) < |\boxf A|$.
\end{lemma}
\begin{proof}
By induction on $|\pi_2|$. The non-trivial case is when $\pi_2$ ends with $\boxf_L$ on $\boxf A$ as shown below left. By induction hypothesis we have $\deriveshallow{\psi'}{A \struct Z_2'[X^l]}$ where $cr(\psi') < |\boxf A|$. The derivation $\pi$ is constructed as shown below right:
$$
\AxiomC{$\psi$}
\noLine \UnaryInfC{$A \struct Z_2'[(\boxf A)^l]$}
\RightLabel{$\boxf_L$} \UnaryInfC{$\boxf A \struct \circwhite{Z_2'[(\boxf A)^l]}$}
\DisplayProof
\qquad
\AxiomC{$\pi_1$}
\noLine \UnaryInfC{$X \struct \circwhiteshort{A}$}
\RightLabel{$rp_\circ$} \UnaryInfC{$\circblackshort{X} \struct A$}
\AxiomC{$\psi'$}
\noLine \UnaryInfC{$A \struct Z_2'[X^l]$}
\RightLabel{$cut$} \BinaryInfC{$\circblackshort{X} \struct Z_2'[X^l]$}
\RightLabel{$rp_\circ$} \UnaryInfC{$X \struct \circwhite{Z_2'[X^l]}$}
\DisplayProof
$$
\end{proof}

\begin{lemma}
\label{lm:cut-boxp}
Suppose $\deriveshallow{\pi_1}{X \struct \circblackshort{A}}$ and
$\deriveshallow{\pi_2}{Z_1[(\boxp A)^k] \struct Z_2[(\boxp A)^l]}$ for
some $k$-hole quasi-negative context $Z_1[\cdots]$ and $l$-hole negative
context $Z_2[\cdots]$, and the cut ranks of $\pi_1$ and $\pi_2$ are
smaller than $|\boxp A|.$ Then there exists $\pi$ such that
$\deriveshallow{\pi}{Z_1[X^k] \struct Z_2[X^l]}$ and $cr(\pi) < |\boxp
A|$.
\end{lemma}
\begin{proof}
By induction on $|\pi_2|$. The non-trivial case is when $\pi_2$ ends with $\boxp_L$ on $\boxp A$ as shown below left. By induction hypothesis we have $\deriveshallow{\psi'}{A \struct Z_2'[X^l]}$ where $cr(\psi') < |\boxp A|$. The derivation $\pi$ is constructed as shown below right:
$$
\AxiomC{$\psi$}
\noLine \UnaryInfC{$A \struct Z_2'[(\boxp A)^l]$}
\RightLabel{$\boxp_L$} \UnaryInfC{$\boxp A \struct \circblack{Z_2'[(\boxp A)^l]}$}
\DisplayProof
\qquad
\AxiomC{$\pi_1$}
\noLine \UnaryInfC{$X \struct \circblackshort{A}$}
\RightLabel{$rp_\bullet$} \UnaryInfC{$\circwhiteshort{X} \struct A$}
\AxiomC{$\psi'$}
\noLine \UnaryInfC{$A \struct Z_2'[X^l]$}
\RightLabel{$cut$} \BinaryInfC{$\circwhiteshort{X} \struct Z_2'[X^l]$}
\RightLabel{$rp_\bullet$} \UnaryInfC{$X \struct \circblack{Z_2'[X^l]}$}
\DisplayProof
$$\end{proof}

\begin{lemma}
\label{lm:cut-diaf}
Suppose $\deriveshallow{\pi_1}{X \struct A}$ and
$\deriveshallow{\pi_2}{Z_1[(\diaf A)^k] \struct Z_2[(\diaf A)^l]}$ for
some $k$-hole quasi-negative context $Z_1[\cdots]$ and $l$-hole negative
context $Z_2[\cdots]$, and the cut ranks of $\pi_1$ and $\pi_2$ are
smaller than $|\diaf A|.$ Then there exists $\pi$ such that
$\deriveshallow{\pi}{Z_1[(\circwhiteshort{X})^k] \struct Z_2[(\circwhiteshort{X})^l]}$ 
and $cr(\pi) < |\diaf A|$.
\end{lemma}
\begin{proof}
By induction on $|\pi_2|$. The non-trivial case is when $\pi_2$ ends with $\diaf_L$ on $\diaf A$ as shown below left. By induction hypothesis we have $\deriveshallow{\psi'}{\circwhiteshort{A} \struct Z_2[(\circwhiteshort{X})^l]}$ where $cr(\psi') < |\diaf A|$. The derivation $\pi$ is constructed as shown below right:
$$
\AxiomC{$\psi$}
\noLine \UnaryInfC{$\circwhiteshort{A} \struct Z_2[(\diaf A)^l]$}
\RightLabel{$\diaf_L$} \UnaryInfC{$\diaf A \struct Z_2[(\diaf A)^l]$}
\DisplayProof
\qquad
\AxiomC{$\pi_1$}
\noLine \UnaryInfC{$X \struct A$}
\AxiomC{$\psi'$}
\noLine  \UnaryInfC{$\circwhiteshort{A} \struct Z_2[(\circwhiteshort{X})^l]$}
\RightLabel{$rp_\bullet$} \UnaryInfC{$A \struct \circblack{Z_2[(\circwhiteshort{X})^l}]$}
\RightLabel{$cut$} \BinaryInfC{$X \struct \circblack {Z_2[(\circwhiteshort{X})^l]}$}
\RightLabel{$rp_\bullet$} \UnaryInfC{$\circwhiteshort{X} \struct Z_2[(\circwhiteshort{X})^l]$}
\DisplayProof
$$
\end{proof}

\begin{lemma}
\label{lm:cut-diap}
Suppose $\deriveshallow{\pi_1}{X \struct A}$ and
$\deriveshallow{\pi_2}{Z_1[(\diap A)^k] \struct Z_2[(\diap A)^l]}$ for
some $k$-hole quasi-negative context $Z_1[\cdots]$ and $l$-hole negative
context $Z_2[\cdots]$, and the cut ranks of $\pi_1$ and $\pi_2$ are
smaller than $|\diap A|.$ Then there exists $\pi$ such that
$\deriveshallow{\pi}{Z_1[(\circblackshort{X})^k] \struct
  Z_2[(\circblackshort{X})^l]}$ and $cr(\pi) < |\diap A|$.
\end{lemma}
\begin{proof}
By induction on $|\pi_2|$. The non-trivial case is when $\pi_2$ ends with $\diap_L$ on $\diap A$ as shown below left. By induction hypothesis we have $\deriveshallow{\psi'}{\circblackshort{A} \struct Z_2[(\circblackshort{X})^l]}$ where $cr(\psi') < |\diap A|$. The derivation $\pi$ is constructed as shown below right:

$$
\AxiomC{$\psi$}
\noLine \UnaryInfC{$\circblackshort{A} \struct Z_2[(\diap A)^l]$}
\RightLabel{$\diap_L$} \UnaryInfC{$\diap A \struct Z_2[(\diap A)^l]$}
\DisplayProof
\qquad
\AxiomC{$\pi_1$}
\noLine \UnaryInfC{$X \struct A$}
\AxiomC{$\psi'$}
\noLine  \UnaryInfC{$\circblackshort{A} \struct Z_2[(\circblackshort{X})^l]$}
\RightLabel{$rp_\circ$} \UnaryInfC{$A \struct \circwhite{Z_2[(\circblackshort{X})^l]}$}
\RightLabel{$cut$} \BinaryInfC{$X \struct \circwhite{Z_2[(\circblackshort{X})^l]}$}
\RightLabel{$rp_\circ$} \UnaryInfC{$\circblackshort{X} \struct Z_2[(\circblackshort{X})^l]$}
\DisplayProof
$$
\end{proof}

\begin{lemma}
\label{lm:cut-non-atm}
Let $A$ be a non-atomic formula. 
Suppose $\deriveshallow{\pi_1}{A, X \struct Y}$ and 
$\deriveshallow{\pi_2}{Z_1[A^k] \struct Z_2[A^l]}$
where $Z_1[\cdots]$ is a $k$-hole positive context, $Z_2[\cdots]$ is
an $l$-hole quasi-positive context, and the cut ranks of $\pi_1$ and $\pi_2$
are smaller than $|A|$. Then there exists $\pi$ such that
$\deriveshallow{\pi}{Z_1[(X \struct Y)^k] \struct Z_2[(X \struct Y)^l]}$ 
and $cr(\pi) < |A|$.
\end{lemma}
\begin{proof}
By induction on $|\pi_2|$ and case analysis on $A.$ The non-trivial
case is when $\pi_2$ ends with a right-introduction rule on $A.$ That
is, in this case, we have $Z_2[A^l] = (Z_2'[A^{l-1}], A$) for some
quasi-positive context $Z_2'[\cdots].$ We distinguish several cases
depending on $A$. We show here the cases where $A$ is either 
$\boxf C$ or $\diaf C$.
\begin{itemize}

\item Suppose $A = \boxf C$ and $\pi_2$ is as below left. 
By induction hypothesis, we have a derivation $\psi'$ of 
$Z_1[(X \struct Y)^k] \struct \circwhiteshort{C}$. 
Then the derivation $\pi$ is constructed as shown below right, 
with $\theta$ obtained by applying Lemma~\ref{lm:cut-boxf} 
to $\psi'$ and $\pi_1$.
$$
\AxiomC{$\psi$}
\noLine \UnaryInfC{$Z_1[(\boxf C)^k] \struct \circwhiteshort{C}$}
\RightLabel{$\boxf_R$} \UnaryInfC{$Z_1[(\boxf C)^k] \struct \boxf C$}
\DisplayProof
\qquad
\AxiomC{$\theta$}
\noLine \UnaryInfC{$Z_1[(X \struct Y)^k], X \struct Y$}
\RightLabel{$\struct_R$} \UnaryInfC{$Z_1[(X \struct Y)^k] \struct X \struct Y$}
\DisplayProof
$$

\item Suppose $A = \diaf C$ and $\pi_2$ is as below left. By induction hypothesis, we have a derivation $\psi'$ of $Z_1'[(X \struct Y)^k] \struct C$. Then the derivation $\pi$ is constructed as shown below right, with $\theta$ obtained by applying Lemma~\ref{lm:cut-diaf} 
to $\psi'$ and $\pi_1$.

$$
\AxiomC{$\psi$}
\noLine \UnaryInfC{$Z_1'[(\diaf C)^k] \struct C$}
\RightLabel{$\diaf_R$} \UnaryInfC{$\circwhite{Z_1'[(\diaf C)^k]} \struct \diaf C$}
\DisplayProof
\qquad
\AxiomC{$\theta$}
\noLine \UnaryInfC{$\circwhite{Z_1'[(X \struct Y)^k]}, X \struct Y$}
\RightLabel{$\struct_R$} \UnaryInfC{$\circwhite{Z_1'[(X \struct Y)^k]} \struct X \struct Y$}
\DisplayProof
$$

\end{itemize}
The other cases are treated analogously, using Lemmas~\ref{lm:cut-or}, \ref{lm:cut-and}, \ref{lm:cut-impl}, \ref{lm:cut-excl}, \ref{lm:cut-boxp} and \ref{lm:cut-diap}.
\end{proof}

\begin{theorem}[Cut elimination for $\lbikt$]
\label{thm:cut-elim}
If $X \struct Y$ is $\lbikt$-derivable then it is also
$\lbikt$-derivable without using cut.
\end{theorem}
\begin{proof}
As typical in cut elimination proofs, we remove topmost cuts in succession. 
Let $\pi$ be a derivation of $\lbikt$ with a topmost cut instance
$$
\AxiomC{$\pi_1$}
\noLine \UnaryInfC{$X_1 \struct Y_1, A$}
\AxiomC{$\pi_2$}
\noLine \UnaryInfC{$A, X_2 \struct Y_2$}
\RightLabel{cut} \BinaryInfC{$X_1, X_2 \struct Y_1, Y_2$}
\DisplayProof
$$

Note that $\pi_1$ and $\pi_2$ are both cut-free since this is a
topmost instance in $\pi.$ We use induction on the size of $A$
to eliminate  this topmost instance of cut.  

If $A$ is an atomic formula $p$ then the cut free derivation is
constructed as follows where $\psi$ is obtained from applying
Lemma~\ref{lm:cut-atm} to $\pi_2$ and $\pi_1$:

$$
\AxiomC{$\psi$}
\noLine \UnaryInfC{$X_1 \struct Y_1, (X_2 \struct Y_2)$}
\RightLabel{$s_R$} \UnaryInfC{$X_1, X_2 \struct Y_1, Y_2$}
\DisplayProof
$$

If $A$ is non-atomic, using Lemma~\ref{lm:cut-non-atm} we get the
following derivation $\pi'$:
$$
\AxiomC{$\psi$}
\noLine \UnaryInfC{$X_1 \struct Y_1, (X_2 \struct Y_2)$}
\RightLabel{$s_R$} \UnaryInfC{$X_1, X_2 \struct Y_1, Y_2$}
\DisplayProof
$$
We have $cr(\pi') < |A|$ by Lemma~\ref{lm:cut-non-atm}, therefore by
induction hypothesis, we can remove all the cuts in $\pi'$ to get a
cut-free derivation of $X_1, X_2 \struct Y_1, Y_2.$
\end{proof}

\section{Equivalence between $\lbikt$ and $\dbikt$}
\label{sec:equiv}

We now show that $\lbikt$ and $\dbikt$ are equivalent.
We first show that every derivation in $\dbikt$ can be mimicked
by a cut-free derivation in $\lbikt.$
The interesting cases involve showing that the propagation rules of
$\dbikt$ are derivable in $\lbikt$ using residuation. This is not 
surprising since the residuation rules in display calculi are used
exactly for the purpose of displaying and un-displaying sub-structures
so that inference rules can be applied to them.

\begin{theorem}
\label{thm:soundness}
For any $X$ and $Y$, if $\derivedeep{\pi}{X \struct Y}$ 
then $\deriveshallow{\pi'}{X \struct Y}$.
\end{theorem}
\begin{proof}
We show that each deep inference rule $\rho$ is derivable in the
shallow system. This is done by case analysis of the context
$\Sigma[~]$ in which the deep rule $\rho$ applies. 
Note that if a deep inference rule $\rho$ is applicable
to $X \struct Y$, then the context $\Sigma[~]$ in this case
is either $[~]$, a positive context or a negative context. 
In the first case, it is easy to show that each valid
instance of $\rho$ where $\Sigma[~] = [~]$ is 
derivable in the shallow system.
For the case where $\Sigma[~]$ is either positive or negative,
we use the display property. We show here the case where
$\rho$ is a rule with a single premise; the other cases
are similar. 
Suppose $\rho$ is as below left. By the display properties, we need only to show that
the rule shown below right is derivable in the shallow system
for some structure $X'$:
$$
\AxiomC{$\Sigma^+[U]$}
\RightLabel{$\rho$} \UnaryInfC{$\Sigma^+[V]$}
\DisplayProof
\qquad
\AxiomC{$X' \struct U$}
\RightLabel{} \UnaryInfC{$X' \struct V$}
\DisplayProof
$$
For example, to show soundness of $\boxf_{L2}$
it is enough to show that the following are derivable:
$$
\AxiomC{$R \struct (\boxf A, X \struct \circwhite{A \struct Y}, Z)$}
\RightLabel{} \UnaryInfC{$R \struct (\boxf A, X \struct \circwhiteshort{Y}, Z)$}
\DisplayProof
\qquad
\AxiomC{$(\boxf A, X \struct \circwhite{A \struct Y}, Z) \struct R$}
\RightLabel{} \UnaryInfC{$(\boxf A, X \struct \circwhiteshort{Y}, Z) \struct R$}
\DisplayProof
$$
Both reduce to showing that the following rule shown below left is derivable; the derivation 
below right gives the required:
$$
\AxiomC{$\boxf A, X \struct \circwhite{A \struct Y}, Z$}
\RightLabel{} \UnaryInfC{$\boxf A, X \struct \circwhiteshort{Y}, Z$}
\DisplayProof
\qquad
\AxiomC{$\boxf A, X \struct \circwhite{A \struct Y}, Z$}
\RightLabel{$\struct_L$} \UnaryInfC{$\boxf A, X \struct Z \struct \circwhite{A \struct Y}$}
\RightLabel{$rp_\circ$} \UnaryInfC{$\circblack{\boxf A, X \struct Z} \struct (A \struct Y)$}
\RightLabel{$s_R$} \UnaryInfC{$A, \circblack{\boxf A, X \struct Z} \struct Y$}
\RightLabel{$\struct_R$} \UnaryInfC{$A \struct \circblack{\boxf A, X \struct Z} \struct Y$}
\RightLabel{$\boxf_L$} \UnaryInfC{$\boxf A \struct \circwhite{\circblack{\boxf A, X \struct Z} \struct Y}$}
\RightLabel{$w_L$} \UnaryInfC{$\boxf A, X \struct W \struct \circwhite{\circblack{\boxf A, X \struct Z} \struct Y}$}
\RightLabel{$rp_\circ$} \UnaryInfC{$\circblack{\boxf A, X \struct Z} \struct \circblack{\boxf A, X \struct Z} \struct Y$}
\RightLabel{$s_R$} \UnaryInfC{$\circblack{\boxf A, X \struct Z}, \circblack{\boxf A, X \struct Z} \struct Y$}
\RightLabel{$c_L$} \UnaryInfC{$\circblack{\boxf A, X \struct Z} \struct Y$}
\RightLabel{$rp_\circ$} \UnaryInfC{$(\boxf A, X \struct Z) \struct \circwhiteshort{Y}$}
\RightLabel{$s_L$} \UnaryInfC{$\boxf A, X \struct \circwhiteshort{Y}, Z$}
\DisplayProof
$$
Below are some other cases, the rest are similar or
easier:
$$
\AxiomC{$Z \struct (X \struct Y, A), A, W$}
\RightLabel{$\struct_{R1}$} \UnaryInfC{$Z \struct (X \struct Y, A), W$}
\DisplayProof
\qquad
\leadsto
\qquad
\AxiomC{$Z \struct (X \struct Y, A), A, W$}
\RightLabel{$\struct_L$} \UnaryInfC{$(Z \struct W) \struct (X \struct Y, A), A$}
\RightLabel{$s_R$} \UnaryInfC{$(Z \struct W), X \struct Y, A, A$}
\RightLabel{$c_R$} \UnaryInfC{$(Z \struct W), X \struct Y, A$}
\RightLabel{$\struct_R$} \UnaryInfC{$(Z \struct W) \struct (X \struct Y, A)$}
\RightLabel{$s_L$} \UnaryInfC{$Z \struct (X \struct Y, A), W$}
\DisplayProof
$$

$$
\AxiomC{$Z, A, \circwhite{\boxp A, X} \struct Y$}
\RightLabel{$\boxp_{L1}$} \UnaryInfC{$Z, \circwhite{\boxp A, X} \struct Y$}
\DisplayProof
\qquad
\leadsto
\qquad
\AxiomC{$Z, A, \circwhite{\boxp A, X} \struct Y$}
\RightLabel{$\struct_R$} \UnaryInfC{$A, \circwhite{\boxp A, X} \struct (Z \struct Y)$}
\RightLabel{$\struct_R$} \UnaryInfC{$A \struct (\circwhite{\boxp A, X} \struct (Z \struct Y))$}
\RightLabel{$\boxp_L$} \UnaryInfC{$\boxp A \struct \circblack{\circwhite{\boxp A, X} \struct (Z \struct Y)}$}
\RightLabel{$w_L$} \UnaryInfC{$\boxp A, X \struct \circblack{\circwhite{\boxp A, X} \struct (Z \struct Y)}$}
\RightLabel{$rp_\bullet$} \UnaryInfC{$\circwhite{\boxp A, X} \struct (\circwhite{\boxp A, X} \struct (Z \struct Y))$}
\RightLabel{$s_R$} \UnaryInfC{$\circwhite{\boxp A, X}, \circwhite{\boxp A, X} \struct (Z \struct Y)$}
\RightLabel{$c_{L}$} \UnaryInfC{$\circwhite{\boxp A, X} \struct (Z \struct Y)$}
\RightLabel{$s_R$} \UnaryInfC{$Z, \circwhite{\boxp A, X} \struct Y$}
\DisplayProof
$$
\end{proof}

We now show that any {\em cut-free} $\lbikt$-derivation
can be transformed into a cut-free $\dbikt$-derivation.
This requires proving cut-free admissibility 
of various structural rules in $\dbikt$.
The admissibility of general weakening and
{\em formula} contraction (not general contraction,
which we will show later) are straightforward
by induction on the height of derivations.

\begin{lemma}[Admissibility of general weakening]\label{lm:adm-weak}
For any structures $X$ and $Y$: if $\derivedeep{\pi}{\Sigma[X]}$ 
then $\derivedeep{\pi'}{\Sigma[X, Y]}$ such that $|\pi'| \leq |\pi|$.
\end{lemma}

\begin{lemma}[Admissibility of formula contraction]\label{lm:adm-contr-fml}
For any structure $X$ and formula $A$: if $\derivedeep{\pi}{\Sigma[X, A, A]}$ 
then $\derivedeep{\pi'}{\Sigma[X, A]}$ such that $|\pi'| \leq |\pi|$.
\end{lemma}

Once weakening and formula contraction are shown admissible, it remains
to show that the residuation rules of $\lbikt$ are also admissible.
In contrast to the case with the deep inference system for bi-intuitionistic 
logic, the combination of modal and intuitionistic structural connectives
complicates the proof of this admissibility. 
It seems crucial to first show ``deep'' admissibility
of certain forms of residuation for $\struct$. 
We state the required lemmas below.

Unless stated otherwise, all lemmas in this section are proved by
induction on $|\pi|$, and $\pi_1'$ is obtained from $\pi_1$ using the
induction hypothesis. 
We label a dashed line with the lemma used to
obtain the conclusion from the premise.

\begin{lemma}[Deep admissibility of structural rules]
\label{lm:deep-adm}
The following statements hold for $\dbikt$:
\begin{enumerate}
\item \label{itm:adm-sl}
{\em Deep admissibility of $s_L$.} 
If $\derivedeep{\pi}{\Sigma[(X \struct Y), Z \struct W]}$ 
then $\derivedeep{\pi'}{\Sigma[X, Z \struct Y, W]}$ such that $|\pi'| \leq |\pi|$.

\item \label{itm:adm-sr} 
{\em Deep admissibility of $s_R$.}
If $\derivedeep{\pi}{\Sigma[X \struct Y, (Z \struct W)]}$ then 
$\derivedeep{\pi'}{\Sigma[X, Z \struct Y, W]}$ such that $|\pi'| \leq |\pi|$.

\item \label{itm:adm-structl}
{\em Deep admissibility of $\struct_L$.}
  If $\derivedeep{\pi}{\Sigma[X \struct Y, Z]}$ and $\Sigma$ is either
  the empty context $[~]$ or a negative context $\Sigma_1^-[~]$, then
  $\derivedeep{\pi'}{\Sigma[(X \struct Y) \struct Z]}$.

\item \label{itm:adm-structr}
{\em Deep admissibility of $\struct_R$.}
  If $\derivedeep{\pi}{\Sigma[X, Y \struct Z]}$ and $\Sigma$ is either
  the empty context $[~]$ or a positive context $\Sigma_1^+[~]$, then
  $\derivedeep{\pi'}{\Sigma[X \struct (Y \struct Z)]}$.
\end{enumerate}
\end{lemma}
\begin{proof}
We prove item (\ref{itm:adm-sl}) and item (\ref{itm:adm-structr});
the other two items can be proved symmetrically. Both are proved by
induction on the height of $\pi.$

\noindent (\ref{itm:adm-sl}): 
The only interesting cases to consider are ones in which $\pi$ ends 
with a propagation rule that moves a formula into or out of 
the structure $(X \struct Y)$, or, one which moves a formula
across from $X$ to $Y$ or vice versa. 
We show here one case for each of these movements, involving the 
propagation rule $\struct_{L1}$ and $\struct_{L2}$. For simplicity, we omit the
context $\Sigma[~]$ as it is not affected by the propagation rule.

Suppose $(X \struct Y) = (A, X' \struct Y)$ and $\pi$ is
as given below left. Then $\pi$ is transformed to 
$\pi'$ (below right). Note that since formula contraction
is height-preserving admissible, we have that 
$|\pi'| \leq |\pi_1'| \leq |\pi_1| < |\pi|.$
$$
\AxiomC{$\pi_1$}
\noLine \UnaryInfC{$(A, X' \struct Y), A, Z \struct W$}
\RightLabel{$\struct_{L1}$} 
\UnaryInfC{$(A, X' \struct Y), Z \struct W$}
\DisplayProof
\qquad
\leadsto
\qquad
\AxiomC{$\pi_1'$}
\noLine \UnaryInfC{$A, X', A, Z \struct Y, W$}
\RightLabel{Lm.~\ref{lm:adm-contr-fml}}
\dashedLine
\UnaryInfC{$A, X', Z \struct Y, W$}
\DisplayProof
$$
Suppose $(X \struct Y) = (X', A \struct (Y_1 \struct Y_2), Y_3)$
and $\pi$ is as given below left. Then $\pi'$ is given below right. 
$$
\AxiomC{$\pi_1$}
\noLine \UnaryInfC{$(X', A \struct (A, Y_1 \struct Y_2),Y_3), Z \struct W$}
\RightLabel{$\struct_{L2}$} \UnaryInfC{$(X', A \struct (Y_1 \struct Y_2), Y_3), Z \struct W$}
\DisplayProof
\qquad
\leadsto
\qquad
\AxiomC{$\pi_1'$}
\noLine \UnaryInfC{$X', A, Z \struct (A, Z, Y_1 \struct Y_2), W$}
\RightLabel{$\struct_{L2}$} \UnaryInfC{$X', A, Z \struct (Y_1 \struct Y_2), Y_3, W$}
\DisplayProof
$$

\noindent (\ref{itm:adm-structr}):  
The only non-trivial cases are when $\pi$ ends with a propagation
rule that moves a formula across the context $\Sigma[~]$ and 
the structure $(X,Y \struct Z)$; or across $(X,Y)$ and $Z.$
We show here one interesting proof transformation involving the latter:
Suppose $(X,Y\struct Z) = (X', \boxp A, Y \struct Z_1, \bullet Z_2)$
and $\pi$ is as given below left. Then $\pi'$ is as given below right.
$$
\AxiomC{$\pi_1$}
\noLine \UnaryInfC{$X', \boxp A, Y \struct Z_1, \circblack{A \struct Z_2}$}
\RightLabel{$\boxp_{L2}$} \UnaryInfC{$X', \boxp A, Y \struct Z_1, \circblackshort{Z_2}$}
\DisplayProof
\ \ 
\leadsto
\ \ 
\AxiomC{$\pi_1'$} 
\noLine \UnaryInfC{$X', \boxp A \struct (Y \struct Z_1, \circblack{A \struct Z_2})$}
\dashedLine \RightLabel{Lm.~\ref{lm:adm-weak}} 
\UnaryInfC{$X', \boxp A \struct (\boxp A, Y \struct Z_1, \circblack{A \struct Z_2})$}
\RightLabel{$\boxp_{L2}$} \UnaryInfC{$X', \boxp A \struct (\boxp A, Y \struct Z_1, \circblackshort{Z_2})$}
\RightLabel{$\struct_{L2}$} \UnaryInfC{$X', \boxp A \struct (Y \struct Z_1, \circblackshort{Z_2})$}
\DisplayProof
$$
\end{proof}

We now show that the residuation rules of $\lbikt$ for $\circ$- and
$\bullet$-structures are admissible in $\dbikt$, i.e., they can be
simulated by the propagation rules of $\dbikt$. 
First we prove a more general admissibility of residuation, which is needed for 
the induction hypothesis of the specific cases. 

\begin{lemma}\label{lm:adm-rpbulletright}
If $\derivedeep{\pi}{X \struct Y, \circblackshort{Z}}$ then 
$\derivedeep{\pi'}{\circwhite{X \struct Y} \struct Z}$.
\end{lemma}
\begin{proof}
The proof is by induction on $|\pi|.$
If $\pi$ ends with a deep inference rule acting on a substructure
inside $X$, $Y$ or $Z$, then $\pi'$ can be constructed straightforwardly
from the induction hypothesis. We look at the more interesting cases
where $\pi$ ends with a propagation rule that moves a formula
across $X$, $Y$ or $Z.$
We give here a representative case; it is not difficult to 
prove the other cases. 
Suppose $X = (X', \boxp A)$ and $\pi$ is the derivation below left.
Then $\pi'$ is constructed as shown below right, where $\pi_1'$
is obtained by applying the induction hypothesis to $\pi_1.$
$$
\AxiomC{$\pi_1$}
\noLine \UnaryInfC{$X', \boxp A  \struct Y, \circblack{A \struct Z}$}
\RightLabel{$\boxp_{L2}$} \UnaryInfC{$X', \boxp A \struct Y, \circblack{Z}$}
\DisplayProof
\qquad
\leadsto
\qquad
\AxiomC{$\pi_1'$}
\noLine \UnaryInfC{$\circwhite{X', \boxp A \struct Y} \struct (A \struct Z)$}
\dashedLine\RightLabel{Lm~\ref{lm:deep-adm}(\ref{itm:adm-sr})} \UnaryInfC{$A, \circwhite{X', \boxp A \struct Y} \struct Z$}
\dashedLine \RightLabel{Lm~\ref{lm:adm-weak}} \UnaryInfC{$A, \circwhite{\boxp A, (X', \boxp A \struct Y)} \struct Z$}
\RightLabel{$\boxp_{L1}$} \UnaryInfC{$\circwhite{\boxp A, (X', \boxp A \struct Y)} \struct Z$}
\RightLabel{$\struct_{L1}$} \UnaryInfC{$\circwhite{X', \boxp A \struct Y} \struct Z$}
\DisplayProof
$$
\end{proof}

\begin{lemma}\label{lm:adm-rpcircleft}
If $\derivedeep{\pi}{\circwhiteshort{X}, Y \struct Z}$ then $\derivedeep{\pi'}{X \struct \circblack{Y \struct Z}}$.
\end{lemma}
\begin{proof}
The proof is by induction on $|\pi|.$
As with the proof of Lemma~\ref{lm:adm-rpbulletright},
the only intersting cases are when $\pi$ ends with a propagation rule
that moves a formula across $X$, $Y$ or $Z.$
We show here one case; the others can be proved similarly.
Suppose $X = (X', \boxp A)$ and $\pi$ is as given below left.
Then $\pi'$ is as given below right, where $\pi_1'$ is obtained
by applying the induction hypothesis to $\pi_1.$
$$
\AxiomC{$\pi_1$}
\noLine \UnaryInfC{$A, \circwhite{X', \boxp A} , Y \struct Z$}
\RightLabel{$\boxp_{L1}$} \UnaryInfC{$\circwhite{X', \boxp A}, Y \struct Z$}
\DisplayProof
\qquad
\leadsto
\qquad
\AxiomC{$\pi_1'$}
\noLine
\UnaryInfC{$X', \boxp A \struct \circblack{A, Y \struct Z)}$}
\dashedLine \RightLabel{Lm~\ref{lm:deep-adm}.\ref{itm:adm-structr}} 
\UnaryInfC{$X', \boxp A \struct \circblack{A \struct (Y \struct Z)}$}
\RightLabel{$\boxp_{L2}$}  \UnaryInfC{$X', \boxp A \struct \circblack{Y \struct Z}$}
\DisplayProof
$$
\end{proof}
The following two lemmas are symmetric to Lemma~\ref{lm:adm-rpbulletright} 
and Lemma~\ref{lm:adm-rpcircleft} and can be proved similarly. 
\begin{lemma}\label{lm:adm-rpcircright}
If $\derivedeep{\pi}{X \struct Y, \circwhiteshort{Z}}$ then $\derivedeep{\pi'}{\circblack{X \struct Y} \struct Z}$.
\end{lemma}

\begin{lemma}\label{lm:adm-rpbulletleft}
If $\derivedeep{\pi}{\circblackshort{X}, Y \struct Z}$ then $\derivedeep{\pi'}{X \struct \circwhite {Y \struct Z}}$.
\end{lemma}

We are now ready to prove the main lemma about admissibility of residuation rules. 
\begin{lemma}
[Admissibility of residuation]
\label{lm:adm-rp}
The following statements hold in $\dbikt$:
\begin{enumerate}

\item 
{\em Admissibility of $rp_\bullet$.}
\label{itm:adm-rpbulletrightspec}
If $\derivedeep{\pi}{X \struct \circblackshort{Z}}$ then $\derivedeep{\pi'}{\circwhiteshort{X} \struct Z}$.

\item \label{itm:adm-rpcircleftspec}
{\em Admissibility of $rp_\bullet$.} 
If $\derivedeep{\pi}{\circwhiteshort{X} \struct Z}$ then 
$\derivedeep{\pi'}{X \struct \circblackshort{Z}}$.

\item\label{itm:adm-rpcircrightspec}
{\em Admissibility of $rp_\circ$.}
If $\derivedeep{\pi}{X \struct \circwhiteshort{Z}}$ then $\derivedeep{\pi'}{\circblackshort{X} \struct Z}$.

\item \label{itm:adm-rpbulletleftspec}
{\em Admissibility of $rp_\circ$.}
If $\derivedeep{\pi}{\circblackshort{X} \struct Z}$ then $\derivedeep{\pi'}{X \struct \circwhiteshort{Z}}$.

\end{enumerate}
\end{lemma}
\begin{proof}
This is straightforward given Lemma~\ref{lm:adm-rpbulletright} -- \ref{lm:adm-rpbulletleft}. 
We outline the proofs for item \ref{itm:adm-rpbulletrightspec} and \ref{itm:adm-rpcircleftspec}; 
the rest can be proved symmetrically. 

\noindent (\ref{itm:adm-rpbulletrightspec}): 
The only interesting cases are when $\pi$ ends with a propagation rule
that moves a formula into or out of $\bullet Z.$
\begin{itemize}
\item Suppose $Z = (Z', \diaf A)$ and $\pi$ is as given below left. 
Lemma~\ref{lm:adm-rpbulletright} gives $\pi_2$ from $\pi_1$, and
$\pi'$ is as below right:
$$
\AxiomC{$\pi_1$}
\noLine \UnaryInfC{$X \struct A, \circblack{Z, \diaf A}$}
\RightLabel{$\diaf_{R1}$} \UnaryInfC{$X \struct \circblack{Z, \diaf A}$}
\DisplayProof
\qquad
\leadsto
\qquad
\AxiomC{$\pi_2$}
\noLine \UnaryInfC{$\circwhite{X \struct A} \struct Z, \diaf A$}
\RightLabel{$\diaf_{R2}$}  \UnaryInfC{$\circwhite{X} \struct Z, \diaf A$}
\DisplayProof
$$

\item Suppose $X = (X', \boxp A)$ and $\pi$ is as given below left. 
Then $\pi'$ is shown below right, where $\pi_1'$ is the result of
applying the induction hypothesis to $\pi_1.$
$$
\AxiomC{$\pi_1$}
\noLine \UnaryInfC{$X', \boxp A  \struct \circblack{A \struct Z}$}
\RightLabel{$\boxp_{L2}$} \UnaryInfC{$X', \boxp A \struct \circblack{Z}$}
\DisplayProof
\qquad
\leadsto
\qquad
\AxiomC{$\pi_1'$}
\noLine \UnaryInfC{$\circwhite{X', \boxp A} \struct (A \struct Z)$}
\dashedLine \RightLabel{Lemma~\ref{lm:deep-adm}(\ref{itm:adm-sr})} \UnaryInfC{$A, \circwhite{X', \boxp A} \struct Z$}
\RightLabel{$\boxp_{L1}$} \UnaryInfC{$\circwhite{X', \boxp A} \struct Z$}
\DisplayProof
$$
\end{itemize}

\noindent (\ref{itm:adm-rpcircleftspec}): 
The only interesting cases are when $\pi$ ends with a propagation rule
that moves a formula into or out of $\circ X.$
\begin{itemize}
\item Suppose $X = (X', \boxp A)$ and $\pi$ is as given below left. 
Lemma~\ref{lm:adm-rpcircleft} gives $\pi_2$ from $\pi_1$, and
$\pi'$ is as below right:
$$
\AxiomC{$\pi_1$}
\noLine \UnaryInfC{$\circ(X',\boxp A), A \struct Z$}
\RightLabel{$\boxp_{L1}$} \UnaryInfC{$\circ(X', \boxp A) \struct Z$}
\DisplayProof
\qquad
\leadsto
\qquad
\AxiomC{$\pi_2$}
\noLine \UnaryInfC{$X', \boxp A \struct \bullet (A \struct Z)$}
\RightLabel{$\boxp_{L2}$}  \UnaryInfC{$X', \boxp A \struct \bullet Z$}
\DisplayProof
$$

\item Suppose $Z = (Z', \diaf A)$ and $\pi$ is as given below left. 
Then $\pi'$ is shown below right, where $\pi_1'$ is the result of
applying the induction hypothesis to $\pi_1.$
$$
\AxiomC{$\pi_1$}
\noLine \UnaryInfC{$\circ (X \struct A)  \struct Z', \diaf A$}
\RightLabel{$\diaf_{R2}$} \UnaryInfC{$\circ X \struct Z', \diaf A$}
\DisplayProof
\qquad
\leadsto
\qquad
\AxiomC{$\pi_1'$}
\noLine 
\UnaryInfC{$(X \struct A) \struct \bullet(Z', \diaf A)$}
\dashedLine \RightLabel{Lemma~\ref{lm:deep-adm}(\ref{itm:adm-sl})} 
\UnaryInfC{$X \struct A, \bullet(Z', \diaf A)$}
\RightLabel{$\diaf_{R1}$} 
\UnaryInfC{$X \struct \bullet (Z', \diaf A)$}
\DisplayProof
$$

\end{itemize}

\end{proof}

\paragraph{Admissibility of general contraction}

To prove the admissibility of structure contraction, we need
to prove several distribution properties among structural connectives. 
These are stated in the following four lemmas.

\begin{lemma}\label{lm:adm-contr-white}
If $\derivedeep{\pi}{\Sigma^+[\circwhite{X \struct Y}, \circwhiteshort{Y}]}$ 
and contraction on structures is admissible
for all derivations $\pi_1$ such that $|\pi_1| \leq |\pi|$ then
$\derivedeep{\pi'}{\Sigma^+[\circwhite{X \struct Y}]}$.
\end{lemma}
\begin{proof}
By induction on the height of $\pi$. The interesting cases are when
$\pi$ ends with a propagation rule that moves a formula into either
$\circwhite{X \struct Y}$ or $\circwhiteshort{Y}$:
\begin{itemize}
\item Suppose $\pi$ ends as below left. Then by
Lemma~\ref{lm:deep-adm}(\ref{itm:adm-sr}), there is a
derivation $\pi_2$ of $\boxf A \struct \circwhite{A, X \struct
Y}, \circwhiteshort{Y}$ such that $|\pi_2| \leq |\pi_1|$. Then
we can apply the induction hypothesis to $\pi_2$ to obtain a
derivation $\pi_3$ of $\boxf A \struct \circwhite{A, X \struct
Y}$. Then the derivation below right gives the required:
$$
\AxiomC{$\pi_1$}
\noLine \UnaryInfC{$\boxf A \struct \circwhite{A \struct (X \struct Y)}, \circwhiteshort{Y}$}
\RightLabel{$\boxf_{L2}$} \UnaryInfC{$\boxf A \struct \circwhite{X \struct Y}, \circwhiteshort{Y}$}
\DisplayProof
\qquad
\AxiomC{$\pi_3$}
\noLine \UnaryInfC{$\boxf A \struct \circwhite{A, X \struct Y}$}
\dashedLine \RightLabel{Lemma~\ref{lm:deep-adm}(\ref{itm:adm-structr})} 
\UnaryInfC{$\boxf A \struct \circwhite{A \struct (X \struct Y)}$}
\RightLabel{$\boxf_{L2}$} \UnaryInfC{$\boxf A \struct \circwhite{X \struct Y}$}
\DisplayProof
$$	
\item Suppose $\pi$ ends as below left. Then applying Lemma~\ref{lm:adm-weak} twice, 
we obtain a derivation $\pi_2$ of 
$\boxf A \struct \circwhite{A, X \struct Y}, \circwhite{A, X \struct Y}$ 
such that $|\pi_2| \leq |\pi_1|$. 
Then we apply the assumption of this lemma to $\pi_2$ to obtain a derivation 
$\pi_3$ of $\boxf A \struct \circwhite{A, X \struct Y}$. 
Then the derivation below right gives the required:
$$
\AxiomC{$\pi_1$}
\noLine \UnaryInfC{$\boxf A \struct \circwhite{X \struct Y}, \circwhite{A \struct Y}$}
\RightLabel{$\boxf_{L2}$} \UnaryInfC{$\boxf A \struct \circwhite{X \struct Y}, \circwhiteshort{Y}$}
\DisplayProof
\qquad
\AxiomC{$\pi_3$}
\noLine \UnaryInfC{$\boxf A \struct \circwhite{A, X \struct Y}$}
\dashedLine \RightLabel{Lemma~\ref{lm:deep-adm}(\ref{itm:adm-structr})} 
\UnaryInfC{$\boxf A \struct \circwhite{A \struct (X \struct Y)}$}
\RightLabel{$\boxf_{L2}$} \UnaryInfC{$\boxf A \struct \circwhite{X \struct Y}$}
\DisplayProof
$$	
\end{itemize}
\end{proof}

\begin{lemma}\label{lm:adm-contr-white-left}
If $\derivedeep{\pi}{\Sigma^-[\circwhite{X \struct Y}, \circwhiteshort{X}]}$  
then $\derivedeep{\pi'}{\Sigma^-[\circwhite{X \struct Y}]}$.
\end{lemma}
\begin{proof}
By induction on the height of $\pi$. The interesting cases are when
$\pi$ ends with a propagation rule that moves a formula into either
$\circwhite{X \struct Y}$ or $\circwhiteshort{X}$:
\begin{itemize}
\item Suppose $\pi$ ends as below left. 
Then by Lemma~\ref{lm:deep-adm}(\ref{itm:adm-sl}), there is a
derivation $\pi_2$ of $\circwhite{X \struct Y,
A}, \circwhiteshort{X} \struct \diaf A$ such that $|\pi_2| \leq
|\pi_1|$. Then we can apply the induction hypothesis to $\pi_2$ to
obtain a derivation $\pi_3$ of $\circwhite{X \struct Y, A} \struct \diaf A$. 
Then the derivation below right gives the required:
$$
\AxiomC{$\pi_1$}
\noLine \UnaryInfC{$\circwhite{(X \struct Y) \struct A}, \circwhiteshort{X} \struct \diaf A$}
\RightLabel{$\diaf_{R2}$} \UnaryInfC{$\circwhite{X \struct Y}, \circwhiteshort{X} \struct \diaf A$}
\DisplayProof
\qquad
\AxiomC{$\pi_3$}
\noLine \UnaryInfC{$\circwhite{X \struct Y, A} \struct \diaf A$}
\dashedLine \RightLabel{Lemma~\ref{lm:deep-adm}(\ref{itm:adm-structl})} 
\UnaryInfC{$\circwhite{(X \struct Y) \struct A} \struct \diaf A$}
\RightLabel{$\diaf_{R2}$} \UnaryInfC{$\circwhite{X \struct Y} \struct \diaf A$}
\DisplayProof
$$		

\item Suppose $\pi$ ends as below left. Then applying
Lemma~\ref{lm:adm-weak} twice, we obtain a derivation $\pi_2$ of
$\circwhite{X \struct Y, A}, \circwhite{X \struct Y, A} \struct \diaf
A$ such that $|\pi_2| \leq |\pi_1|$. Then we apply the assumption of
this lemma to $\pi_2$ to obtain a derivation $\pi_3$ of
$\circwhite{X \struct Y, A} \struct \diaf A$. Then the derivation
below right gives the required:
$$
\AxiomC{$\pi_1$}
\noLine \UnaryInfC{$\circwhite{X \struct Y}, \circwhite{X \struct A} \struct \diaf A$}
\RightLabel{$\diaf_{R2}$} 
\UnaryInfC{$\circwhite{X \struct Y}, \circwhiteshort{X} \struct \diaf A$}
\DisplayProof
\qquad
\AxiomC{$\pi_3$}
\noLine \UnaryInfC{$\circwhite{X \struct Y, A}\struct \diaf A$}
\dashedLine \RightLabel{Lemma~\ref{lm:deep-adm}(\ref{itm:adm-structl})} 
\UnaryInfC{$\circwhite{(X \struct Y) \struct A} \struct \diaf A$}
\RightLabel{$\diaf_{R2}$} \UnaryInfC{$\circwhite{X \struct Y}  \struct \diaf A$}
\DisplayProof
$$	
\end{itemize}
\end{proof}

\begin{lemma}\label{lm:adm-contr-black}
If $\derivedeep{\pi}{\Sigma^+[\circblack{X \struct Y}, \circblack{Y}]}$  
then $\derivedeep{\pi'}{\Sigma^+[\circblack{X \struct Y}]}$.
\end{lemma}

\begin{lemma}\label{lm:adm-contr-black-left}
If $\derivedeep{\pi}{\Sigma^-[\circblack{X \struct Y}, \circblack{X}]}$  
then $\derivedeep{\pi'}{\Sigma^-[\circblack{X \struct Y}]}$.
\end{lemma}

\begin{lemma}[Admissibility of general contraction]
\label{lm:adm-contr-gen}
For any structure $Y$: if $\derivedeep{\pi}{\Sigma[Y, Y]}$ then
$\derivedeep{\pi'}{\Sigma[Y]}$.
\end{lemma}
\begin{proof}
By induction on the size of $Y$, with a sub-induction on $|\pi|$.
\begin{itemize}
\item For the base case, use Lemma~\ref{lm:adm-contr-fml}. 
\item For the case where $Y$ is a $\struct$-structure, we show 
the sub-case where $Y$ in a negative context, the other case is symmetric:
$$
\AxiomC{$\Sigma[(Y_1 \struct Y_2), (Y_1 \struct Y_2) \struct Z]$}
\RightLabel{Lemma~\ref{lm:deep-adm}(\ref{itm:adm-sl})} \dashedLine \UnaryInfC{$\Sigma[Y_1, (Y_1 \struct Y_2) \struct Y_2, Z]$}
\RightLabel{Lemma~\ref{lm:deep-adm}(\ref{itm:adm-sl})} \dashedLine \UnaryInfC{$\Sigma[Y_1, Y_1 \struct Y_2, Y_2, Z]$}
\RightLabel{IH} \dashedLine \UnaryInfC{$\Sigma[Y_1 \struct Y_2, Y_2, Z]$}
\RightLabel{IH} \dashedLine \UnaryInfC{$\Sigma[Y_1 \struct Y_2, Z]$}
\RightLabel{Lemma~\ref{lm:deep-adm}(\ref{itm:adm-structl})} \dashedLine \UnaryInfC{$\Sigma[(Y_1 \struct Y_2) \struct Z]$}
\DisplayProof
$$

\item For the case where $Y$ is a $\circ-$ or $\bullet$-structure 
and $\pi$ ends with a propagation rule applied to $Y$, there are three non-trivial sub-cases:
\begin{itemize}
\item A formula is propagated into $Y$ and $Y$ is in a
positive context, as below left. Then by Lemma~\ref{lm:adm-contr-white}, 
there is a derivation $\pi_1'$ of $\boxf A, X \struct \circwhite{A \struct Z}$. 
Then the derivation below right gives the required:
$$
\AxiomC{$\pi_1$}
\noLine \UnaryInfC{$\boxf A, X \struct \circwhite{A \struct Z}, \circwhiteshort{Z}$}
\RightLabel{$\boxf_{L2}$} \UnaryInfC{$\boxf A, X \struct \circwhiteshort{Z}, \circwhiteshort{Z}$}
\DisplayProof
\qquad
\AxiomC{$\pi_1'$}
\noLine \UnaryInfC{$\boxf A, X \struct \circwhite{A \struct Z}$}
\RightLabel{$\boxf_{L2}$} \UnaryInfC{$\boxf A, X \struct \circwhiteshort{Z}$}
\DisplayProof
$$
\item A formula is propagated into $Y$ and $Y$ is in a
negative context, as below left. Then by Lemma~\ref{lm:adm-contr-white-left}, 
there is a derivation $\pi_1'$ of $\circwhite{Z \struct A} \struct X, \diaf A$. 
Then the derivation below right gives the required:
$$
\AxiomC{$\pi_1$}
\noLine \UnaryInfC{$\circwhite{Z \struct A}, \circwhiteshort{Z} \struct X, \diaf A$}
\RightLabel{$\diaf_{R2}$} \UnaryInfC{$\circwhiteshort{Z}, \circwhiteshort{Z} \struct X, \diaf A$}
\DisplayProof
\qquad
\AxiomC{$\pi_1'$}
\noLine \UnaryInfC{$\circwhite{Z \struct A} \struct X, \diaf A$}
\RightLabel{$\diaf_{R2}$} \UnaryInfC{$\circwhiteshort{Z} \struct X, \diaf A$}
\DisplayProof
$$
\item A formula is propagated out of $Y$, as below left. 
In this case we use the sub-induction hypothesis to obtain a 
derivation $\pi_1'$ of $X \struct A, \circwhite{\diap A, Z}$. 
Then the derivation below right gives the required:
$$
\AxiomC{$\pi_1$}
\noLine \UnaryInfC{$X \struct A, \circwhite{\diap A, Z}, \circwhite{\diap A, Z}$}
\RightLabel{$\diap_{R1}$} \UnaryInfC{$X \struct \circwhite{\diap A, Z}, \circwhite{\diap A, Z}$}
\DisplayProof
\qquad
\AxiomC{$\pi_1'$}
\noLine \UnaryInfC{$X \struct A, \circwhite{\diap A, Z}$}
\RightLabel{$\diap_{R1}$} \UnaryInfC{$X \struct \circwhite{\diap A, Z}$}
\DisplayProof
$$
\end{itemize}
\end{itemize}
\end{proof}

Once all structural rules of $\lbikt$ are shown
admissible in $\dbikt$, it is easy to show that
every derivation in $\lbikt$ can be translated to
a derivation in $\dbikt$. 

\begin{theorem}\label{thm:completeness}
  For any $X$ and $Y$, if $\deriveshallow{\pi}{X
    \struct Y}$ then $\derivedeep{\pi'}{X \struct Y}$.
\end{theorem}
\begin{proof}
By induction on $|\pi|$, where $\pi_1'$ ($\pi_2'$) is obtained from
$\pi_1$ ($\pi_2$) using the IH. We show some cases where $\pi$ ends in
logical rule applications and some where $\pi$ ends in structural rule
applications. The other interesting cases are similar and use
Lemmas~\ref{lm:deep-adm} and Lemmas~\ref{lm:adm-rp}.
$$
\begin{array}{l}
\AxiomC{$\pi_1$}
\noLine \UnaryInfC{$X \struct A, Y$}
\AxiomC{$\pi_2$}
\noLine \UnaryInfC{$X, B \struct Y$}
\RightLabel{$\impl_L$} \BinaryInfC{$X, A \impl B \struct Y$}
\DisplayProof
\qquad
\leadsto \\
\qquad
\AxiomC{$\pi_1''$}
\noLine \UnaryInfC{$X \struct A, Y$}
\dashedLine \RightLabel{Lm.~\ref{lm:adm-weak}} \UnaryInfC{$X, A \impl B \struct A, Y$}
\dashedLine \RightLabel{Lm.~\ref{lm:deep-adm}(\ref{itm:adm-structl})} \UnaryInfC{$(X, A \impl B \struct A) \struct Y$}
\AxiomC{$\pi_2''$}
\noLine \UnaryInfC{$X, B \struct Y$}
\dashedLine \RightLabel{Lm.~\ref{lm:adm-weak}} \UnaryInfC{$X, A \impl B, B \struct Y$}
\RightLabel{$\impl_L$} \BinaryInfC{$X, A \impl B \struct Y$}
\DisplayProof
\end{array}
$$	

$$
\AxiomC{$\pi_1$}
\noLine \UnaryInfC{$X, Y \struct Z$}
\RightLabel{$\struct_R$} \UnaryInfC{$X \struct (Y \struct Z)$}
\DisplayProof
\qquad
\leadsto
\qquad
\AxiomC{$\pi_1'$}
\noLine \UnaryInfC{$X, Y \struct Z$}
\dashedLine \RightLabel{Lemma~\ref{lm:deep-adm}(\ref{itm:adm-structr})} \UnaryInfC{$X \struct (Y \struct Z)$}
\DisplayProof
$$

$$
\AxiomC{$\pi_1$}
\noLine \UnaryInfC{$\circblackshort{X} \struct Y$}
\RightLabel{$rp_\circ$} \UnaryInfC{$X \struct \circwhiteshort{Y}$}
\DisplayProof
\qquad
\leadsto
\qquad
\AxiomC{$\pi_1'$}
\noLine \UnaryInfC{$\circblackshort{X} \struct Y$}
\dashedLine \RightLabel{Lemma~\ref{lm:adm-rp}.\ref{itm:adm-rpbulletleftspec}} 
\UnaryInfC{$X \struct \circwhiteshort{Y}$}
\DisplayProof
$$

$$
\AxiomC{$\pi_1$}
\noLine \UnaryInfC{$A \struct X$}
\RightLabel{$\boxf_L$} \UnaryInfC{$\boxf A \struct \circwhiteshort{X}$}
\DisplayProof
\qquad
\leadsto
\qquad
\AxiomC{$\pi_1'$}
\noLine \UnaryInfC{$A \struct X$}
\dashedLine \RightLabel{Lm.~\ref{lm:adm-weak}} 
\UnaryInfC{$A, \circblack{\boxf A, (\boxf A \struct \circwhiteshort{X})} \struct X$}
\RightLabel{$\boxf_{L1}$} \UnaryInfC{$\circblack{\boxf A, (\boxf A \struct \circwhiteshort{X})} \struct X$}
\RightLabel{$\struct_{L1}$} \UnaryInfC{$\circblack{\boxf A \struct \circwhiteshort{X}} \struct X$} 
\dashedLine \RightLabel{Lm.~\ref{lm:adm-rp}(\ref{itm:adm-rpbulletleftspec})} 
\UnaryInfC{$(\boxf A \struct \circwhiteshort{X}) \struct \circwhiteshort{X}$}
\dashedLine \RightLabel{Lm.~\ref{lm:deep-adm}(\ref{itm:adm-sl})} 
\UnaryInfC{$\boxf A  \struct \circwhiteshort{X}, \circwhiteshort{X}$}
\dashedLine \RightLabel{Lm.~\ref{lm:adm-contr-gen}} \UnaryInfC{$\boxf A  \struct \circwhiteshort{X}$}
\DisplayProof
$$
\end{proof}

\begin{theorem}\label{thm:equiv}
  For any $X$ and $Y$, $\deriveshallow{\pi}{X \struct Y}$ 
  if and only if~~$\derivedeep{\pi'}{X \struct Y}$.
\end{theorem}
\begin{proof}
By Theorems~\ref{thm:soundness} and~\ref{thm:completeness}. 
\end{proof}

\section{Proof Search}\label{sec:search}

In this section we outline a proof search strategy for $\dbikt$,
closely following the approaches presented in~\cite{postniece2009} and
~\cite{gorepostniecetiu2009}. Here we emphasize the aspects that are
new/different because of the interaction between the tense structures
$\circwhiteshort$ and $\circblackshort$ and the intuitionistic
structure $\struct$. 

Our proof search strategy proceeds in three stages: saturation,
propagation and {\em realisation}.
The saturation phase applies
the ``static rules'' (i.e. those that do not create extra
structural connectives) until further application do not 
lead to any progress. The propagation phase propagates formulaes
across different structural connectives, while the realisation
phase applies the ``dynamic rules'' (i.e., those that create
new structural connectives, e.g., $\impl_R$).

A context $\Sigma[~]$ is said to be {\em headed by a structural
  connective \#} if the topmost symbol in the formation tree of
$\Sigma[~]$ is $\#.$ A context $\Sigma[~]$ is said to be a {\em
  factor} of $\Sigma'[~]$ if $\Sigma[~]$ is a subcontext of
$\Sigma'[~]$ and $\Sigma[~]$ is headed by either $\struct$, $\circ$ or
$\bullet.$ We write $\upcontext{\Sigma}[~]$ to denote the least
factor of $\Sigma[~].$ We write $\upcontext{\Sigma}[X]$ to denote the
structure $\Sigma_1[X]$, if $\Sigma_1[~] = \upcontext{\Sigma}[~].$ We
define the {\em top-level} formulae of a structure as: 
$$
\toplevel{X} =
\{ A \mid X = (A, Y) \text{ for some } A \text{ and } Y \}.
$$ 
For example, if $\Sigma[] = (A, B \struct C, \circblack{D, (E \struct F)
  \struct []})$, then $\upcontext{\Sigma}[G] = \circblackshort(D, (E
\struct F) \struct G)$, and $\toplevel{D, (E \struct F)} = \{D \}$.

Let $\dimpl_{L1}$ and $\impl_{R1}$ denote two new derived rules (see~\cite{postniece2009} for their derivation):
$$
\AxiomC{$\Sigma^-[A, A \dimpl B]$}
\RightLabel{$\dimpl_{L1}$} \UnaryInfC{$\Sigma^-[A \dimpl B]$}
\DisplayProof
\qquad
\qquad
\qquad
\AxiomC{$\Sigma^+[A \impl B, B]$}
\RightLabel{$\impl_{R1}$} \UnaryInfC{$\Sigma^+[A \impl B]$}
\DisplayProof
$$

We now define a notion of a {\em saturated structure},
which is similar to that of a traditional sequent. 
Note that we need to define it for both structures headed by
$\struct$ and those headed by $\circ$ or $\bullet.$
A structure $X \struct Y$ is {\em saturated} if 
it satisfies the following:
\begin{description}
\item[(1)] $\toplevel{X} \cap \toplevel{Y} = \emptyset$
\item[(2)] If $A \land B \in \toplevel{X}$ then $A \in \toplevel{X}$ and $B \in \toplevel{X}$
\item[(3)] If $A \land B \in \toplevel{Y}$ then $A \in \toplevel{Y}$ or $B \in \toplevel{Y}$
\item[(4)] If $A \lor B \in \toplevel{X}$ then $A \in \toplevel{X}$ or $B \in \toplevel{X}$
\item[(5)] If $A \lor B \in \toplevel{Y}$ then $A \in \toplevel{Y}$ and $B \in \toplevel{Y}$
\item[(6)] If $A \impl B \in \toplevel{X}$ then $A \in \toplevel{Y}$ or $B \in \toplevel{X}$
\item[(7)] If $A \dimpl B \in \toplevel{Y}$ then $A \in \toplevel{Y}$ or $B \in \toplevel{X}$
\item[(8)] If $A \dimpl B \in \toplevel{X}$ then $A \in \toplevel{X}$
\item[(9)] If $A \impl B \in \toplevel{Y}$ then $B \in \toplevel{Y}$
\end{description}
For structures of the form $\circ X$ or $\bullet X$, we need to
define two notions of saturation, {\em left saturation} and {\em right saturation}.
The former is used when $\circ X$ is nested in a negative context,
and the latter when it is in a positive context. 
A structure $\circ X$ or $\bullet X$ is {\em left-saturated} if
it satisfies (2), (4), (8) above, and 
\begin{description}
\item[6'] If $A \impl B \in \toplevel{X}$ then $B \in \toplevel{X}.$
\end{description}
Dually, $\circ Y$ or $\bullet Y$ is {\em right-saturated} if 
it satisfies (3), (5), (9) above, and 
\begin{description}
\item[7'] If $A \dimpl B \in \toplevel{Y}$ then $A \in \toplevel{Y}.$
\end{description}

We define structure membership for any two structures $X$ and $Y$ as
follows: $X \in Y$ iff $Y = X, X'$ for some $X'$, modulo associativity
and commutativity of comma. For example, $(A \struct B) \in (A, (A
\struct B), \circwhiteshort{C})$. 
The {\em realisation of formulae} by a structure $X$ is defined
as follows:
\begin{itemize}
\item $A \impl B$ ($A \dimpl B$, resp.) is right-realised
  (resp. left-realised) by $X$ iff there exists  $Z \struct
  W \in X$ such that $A \in \toplevel{Z}$ and $B \in \toplevel{W}$.

\item $\boxf A$ ($\diaf A$ resp.) is right-realised
  (resp. left-realised) by $X$ iff there exists 
  $\circwhite{Z \struct W} \in X$ or $\circwhiteshort{W} \in X$
  (resp. $\circwhite{W \struct Z} \in X$ or $\circwhiteshort{W} \in
  X$) such that $A \in \toplevel{W}$.

\item $\boxp A$ ($\diap A$ resp.) is right-realised
  (resp. left-realised) by $X$ iff there exists 
  $\circblack{W \struct Z} \in X$ or $\circblackshort{Z} \in X$
  (resp.\ $\circblack{Z \struct W} \in X$ or $\circblackshort{Z} \in X$) such
  that $A \in \toplevel{Z}$.
\end{itemize}

We say that a structure $X$ is left-realised 
iff every formula in $\toplevel{X}$ with top-level connective
$\dimpl$, $\diaf$ or $\diap$ is left-realised by $X.$
Right-realisation of $X$ is defined dually. 
We say that a structure occurrence $X$ in $\Sigma[X]$ 
is {\em propagated} iff no propagation rules are
(backwards) applicable to any formula occurrences in $X$.
We define the super-set relation on structures as follows: 
\begin{itemize}
\item $X_1 \struct Y_1 \supset X_0 \struct Y_0$ iff $\toplevel{X_1} \supset \toplevel{X_0}$ or 
$\toplevel{Y_1} \supset \toplevel{Y_0}$.
\item $\circ X \supset \circ Y$ iff $\bullet X \supset \bullet Y$
iff $\toplevel{ X } \supset \toplevel{ Y }.$
\end{itemize}

To simplify presentation, we use the following terminology:
Given a structure $\Sigma[A]$, we say that 
$\upcontext{\Sigma}[A]$ is saturated if $\upcontext{\Sigma}[A]$ is
$X \struct Y$ and it is saturated; or $\upcontext{\Sigma}[A]$
is either $\circ X$ or $\bullet X$ and it is either left- or right-saturated
(depending on its position in $\Sigma[A]$).
We say that $\upcontext{\Sigma}[A]$ is propagated if 
its occurrence in $\Sigma[A]$ is propagated, and
we say that $A$ is realised by $\upcontext{\Sigma}[A]$, if either 
\begin{itemize}
\item $\upcontext{\Sigma}[A] = (X \struct Y)$ and 
either $A \in \toplevel{X}$ is left-realised by $X$, or $A \in \toplevel{Y}$
is right realised by $Y$; or
\item $\upcontext{\Sigma}[A]$ is either $\circ X$ or $\bullet X$,
and, depending on the polarity of $\Sigma[~]$, $A$ is either
left- or right-realised by $X.$
\end{itemize} 

We now outline an approach to proof search in $\dbikt$. 
We approach this by modifying $\dbikt$ to obtain a calculus 
$\dbikts$ that is more amenable to proof search. Our approach follows that
of our previous work on bi-intuitionistic logic~\cite{postniece2009}
since we define syntactic restrictions on rules to enforce a search
strategy. For example, we stipulate that a structure must be saturated
and propagated before child structures can be created using the
$\impl_R$ rule (see condition~\ref{lab:dbikt:two} of
Definition~\ref{def:dbikts}).
Additionally and more importantly, our proof search calculus addresses
the issue that some modal propagation rules of $\dbikt$,
e.g. $\boxf_{L2}$, create $\struct$-structures during backward proof
search. This property of $\dbikt$ is undesirable and gives rise to
non-termination if rules $\boxf_{L2}$ like are applied naively.

\begin{definition}\label{def:dbikts}
Let $\dbikts$ be the system obtained from $\dbikt$ with the following changes:
\begin{enumerate}
\item\label{lab:dbikt:one} Add the derived rules $\dimpl_{L1}$ and $\impl_{R1}$.

\item\label{lab:dbikt:two} Restrict rules $\dimpl_L$, $\impl_R$ with the
  following condition:  the rule is applicable only if
  $\upcontext{\Sigma}[A\#B ]$ is saturated and propagated, and $A\#B
  \text{ is not realised by} \upcontext{\Sigma}[A\#B ]$, for $\# \in
       {\impl, \dimpl}$.

\item\label{lab:dbikt:threeone} Replace rules $\struct_{L1}$ and $\struct_{R1}$
  with the following: 
$$
\AxiomC{$\Sigma[A, (A, X \struct Y), W \struct Z]$}
\RightLabel{$\struct_{L1}$} \UnaryInfC{$\Sigma[(A, X \struct Y), W \struct Z]$}
\DisplayProof
\qquad
\AxiomC{$\Sigma[W \struct Z, (X \struct Y, A), A]$}
\RightLabel{$\struct_{R1}$} \UnaryInfC{$\Sigma[W \struct Z, (X \struct Y, A)]$}
\DisplayProof
$$

\item\label{lab:dbikt:threetwo} Restrict rules $\struct_{L2}$ and $\struct_{R2}$
  with the following condition: the rule is applicable only if $A
  \not\in \toplevel{Y}$.

\item\label{lab:dbikt:four} Replace rules $\diaf_L$, $\boxf_R$, $\diap_L$,
  $\boxp_R$ with the following, where the rule is applicable only
  if $\upcontext{\Sigma}[\#A ]$ is saturated and propagated and $\#A
  \text{ is not realised by} \upcontext{\Sigma}[\#A ]$, for $\# \in
       {\diaf, \boxf, \diap, \boxp}$:
$$       
\AxiomC{$\Sigma^-[\diaf A, \circwhite{A \struct \emptyset}]$}
\RightLabel{$\diaf_L$} \UnaryInfC{$\Sigma^-[\diaf A]$}
\DisplayProof
\qquad
\AxiomC{$\Sigma^+[\boxf A, \circwhite{\emptyset \struct A}]$}
\RightLabel{$\boxf_R$} \UnaryInfC{$\Sigma^+[\boxf A]$}
\DisplayProof
$$
$$
\AxiomC{$\Sigma^-[\diap A, \circblack{A \struct \emptyset}]$}
\RightLabel{$\diap_L$} \UnaryInfC{$\Sigma^-[\diap A]$}
\DisplayProof
\qquad
\AxiomC{$\Sigma^+[\boxp A, \circblack{\emptyset \struct A}]$}
\RightLabel{$\boxp_R$} \UnaryInfC{$\Sigma^+[\boxp A]$}
\DisplayProof
$$

\item\label{lab:dbikt:five} Replace rules $\boxp_{L2}, \boxf_{L2}$ with the
  following, where $A \not\in \toplevel{Y_1}$:
$$
\AxiomC{$\Sigma[\boxp A,X \struct \circblack{A, Y_1 \struct Y_2},Z]$}
\RightLabel{$\boxp_{L2}$} \UnaryInfC{$\Sigma[\boxp A,X \struct \circblack{Y_1 \struct Y_2}, Z]$}
\DisplayProof
\qquad
\AxiomC{$\Sigma[\boxf A, X \struct \circwhite{A, Y_1 \struct Y_2}, Z]$}
\RightLabel{$\boxf_{L2}$} \UnaryInfC{$\Sigma[\boxf A, X \struct \circwhite{Y_1 \struct Y_2}, Z]$}
\DisplayProof
$$

\item\label{lab:dbikt:six} Replace rules $\diaf_{R2}, \diap_{R2}$ with the following, where $A \not\in \toplevel{X_2}$:
$$
\AxiomC{$\Sigma[\circwhite{X_1 \struct X_2, A}, Y \struct Z, \diaf A]$}
\RightLabel{$\diaf_{R2}$} \UnaryInfC{$\Sigma[\circwhite{X_1 \struct X_2}, Y \struct Z, \diaf A]$}
\DisplayProof
\qquad
\AxiomC{$\Sigma[\circblack{X_1 \struct X_2, A}, Y \struct Z, \diap A]$}
\RightLabel{$\diap_{R2}$} \UnaryInfC{$\Sigma[\circblack{X_1 \struct X_2}, Y \struct Z, \diap A]$}
\DisplayProof
$$

\item\label{lab:dbikt:seven} Replace rules $\boxf_{L1}, \diaf_{R1}$, $\boxp_{L1}, \diap_{R1}$ with the following:
$$
\AxiomC{$\Sigma^-[A, \circblack{\boxf A,X \struct Y}]$}
\RightLabel{$\boxf_{L1}$} \UnaryInfC{$\Sigma^-[\circblack{\boxf A, X \struct Y}]$}
\DisplayProof
\qquad
\AxiomC{$\Sigma^+[A, \circblack{Y \struct \diaf A, X}]$}
\RightLabel{$\diaf_{R1}$} \UnaryInfC{$\Sigma^+[\circblack{Y \struct \diaf A, X}]$}
\DisplayProof
$$
$$
\AxiomC{$\Sigma^-[A, \circwhite{\boxp A,X \struct Y}]$}
\RightLabel{$\boxp_{L1}$} \UnaryInfC{$\Sigma^-[\circwhite{\boxp A,X \struct Y}]$}
\DisplayProof
\qquad
\AxiomC{$\Sigma^+[A, \circwhite{Y \struct \diap A,X}]$}
\RightLabel{$\diap_{R1}$} \UnaryInfC{$\Sigma^+[\circwhite{Y \struct \diap A,X}]$}
\DisplayProof
$$

\item\label{lab:dbikt:eight} Replace rules $\impl_L, \dimpl_R$ with the following:
$$       
\AxiomC{$\Sigma[X, A \impl B \struct A, Y]$}
\AxiomC{$\Sigma[X, A \impl B, B \struct Y]$}
\RightLabel{$\impl_L$} \BinaryInfC{$\Sigma[X, A \impl B \struct Y]$}
\DisplayProof
$$
$$
\AxiomC{$\Sigma[X \struct Y, A \dimpl B, A]$}
\AxiomC{$\Sigma[X, B \struct Y, A \dimpl B]$}
\RightLabel{$\dimpl_R$} \BinaryInfC{$\Sigma[X \struct Y, A \dimpl B]$}
\DisplayProof
$$

\item\label{lab:dbikt:nine} Restrict rules $\impl_L$, $\dimpl_R$, $\struct_{L1}$,
  $\struct_{R1}$, $\land_L$, $\land_R$, $\lor_L$, $\lor_R$ and all
  modal propagation rules to the following:
 Let $\Sigma[X_0]$ be the conclusion of the rule and let $\Sigma[X_1]$
  (and $\Sigma[X_2]$) be the premise(s). The rule is applicable only if:
  $\upcontext{\Sigma}[X_1] \supset \upcontext{\Sigma}[X_0]$ and
  $\upcontext{\Sigma}[X_2] \supset \upcontext{\Sigma}[X_0]$.

\end{enumerate}
\end{definition}

We conjecture that $\dbikt$ and $\dbikts$ are equivalent and that
backward proof search in $\dbikt$ terminates. Note that by equivalence
here we mean that $\dbikt$ and $\dbikt$ proves the same set of formulae,
but not necessarily the same set of structures. This is because the
propagation rules in $\dbikts$ are more restricted so as to allow
for easier termination checking. For example, the structure
$A \struct \bullet(\diaf A)$ is derivable in $\dbikt$ but not
in $\dbikts$, although its formula translate is derivable in both
systems.
It is likely that a combination of the techniques 
from~\cite{postniece2009} and~\cite{gorepostniecetiu2009} can be used to prove
termination of proof search in $\dbikts$, given its similarities to the deep inference
systems used in those two works.

\section{Semantics}
\label{sec:semantics}

We now give a Kripke-style semantics for $\bikt$ and show that
$\lbikt$ is sound with respect to the semantics. 
Our semantics for $\bikt$ extend Rauszer's \cite{Rauszer80Phd}
Kripke-style semantics for $\biint$ by clauses for the tense logic
connectives.  We use the classical first-order meta-level connectives
\&, ``or'', ``not'', $\Rightarrow$, $\forall$ and $\exists$ to state
our semantics. 

A Kripke {\em frame} is a tuple $\langle W, \leq, R_\diaf, R_\boxf \rangle$
where $W$ is a non-empty set of worlds and $\leq \ \subseteq (W \times W)$ 
is a reflexive and transitive binary relation
over $W$, and each of $R_\diaf$ and $R_\boxf$ are arbitrary binary
relations over $W$ with the following {\em frame conditions}:
\begin{description}
\item[\rm F1$\diaf$] if $x \leq y \ \&\ x R_\diaf z$ then $\exists w.~y
  R_\diaf w \ \&\ z \leq w$
\item[\rm F2$\boxf$] if $x R_\boxf y \ \&\ y \leq z$ then  $\exists w.~x \leq
  w \ \& \ w R_\boxf z$.
\end{description}

A Kripke {\em model} extends 
a Kripke frame with a mapping  $V$ from $Atoms$ to $2^{W}$
obeying
{\em persistence:}
$$\forall v \geq w.~w \in V(p) \Rightarrow v \in V(p).$$
Given a model $\langle W, \leq, R_\diaf, R_\boxf, V \rangle$, we say that
$w \in W$ {\em satisfies} $p$ if $w \in V(p)$, and write this as $w
\Vdash p$. 
We write $w \not\Vdash p$ to mean $(not)(w \Vdash p)$; that
is, $\exists v \geq w.~v \not\in V(p)$. The relation $\Vdash$
is then extended to formulae as given in Figure~\ref{fig:semantics}.
A $\bikt$-formula $A$ is {\em $\bikt$-valid} if it is
satisfied by every world in every Kripke model. A nested sequent $X
\struct Y$ is $\bikt$-valid if its formula translation is
$\bikt$-valid.

\begin{figure}[t]
\begin{center}
\mysmall{
\begin{tabular}
      {l@{\hspace{0.3cm}}c@{\hspace{0.3cm}}l@{\hspace{1cm}}
       l@{\hspace{0.3cm}}c@{\hspace{0.3cm}}l}
 $w \Vdash \top$      &    & for every $w \in W$ &
 $w \Vdash \bot$      &    & for no $w \in W$ \\
 $w \Vdash A \land B$  & if & $w \Vdash A \ \& \ w \Vdash B$ &
 $w \Vdash A \lor B$   & if & $w \Vdash A$ or $w \Vdash B$ \\
 $w \Vdash A \impl B$  & if & $\forall v \geq w.~v \Vdash A \Rightarrow v \Vdash B$ &
 $w \Vdash \neg A$ & if & $\forall v \geq w.~ v\not\Vdash A$ \\
 $w \Vdash A \dimpl B$ & if & $\exists v \leq w.~v \Vdash A \ \& \ v \not\Vdash B$ &
 $w \Vdash \ \dneg A$ & if & $\exists v \leq w.~v \not\Vdash A$ \\
 $w \Vdash \diaf A$ & if & $\exists v.~w R_\diaf v \ \&\ v \Vdash A$ &
 $w \Vdash \boxf A$ & if & $\forall z.\forall v.~w \leq z ~\&\ z R_\boxf v \Rightarrow v \Vdash A$ \\
 $w \Vdash \diap A$ & if & $\exists v.~w R_\boxf^{-1} v \ \&\ v \Vdash A$ &
 $w \Vdash \boxp A$ & if & $\forall z.\forall v.~w \leq z \ \&\  z R_\diaf^{-1} v \Rightarrow v \Vdash A$ 
\end{tabular}
}
    \caption{Semantics for $\bikt$}
    \label{fig:semantics}
\end{center}
\end{figure}

Our semantics differ from those of
Simpson~\cite{simpson1994} and Ewald~\cite{ewald1986} because we use
two modal accessibility relations instead of one. 
In our calculi, there is no direct relationship between $\diaf$ and
$\boxf$ (or $\diap$ and $\boxp$),
but 
$\diaf$ and $\boxp$ are a residuated pair, as are $\diap$ and
$\boxf$. Semantically, this corresponds to $R_\diap =
R_\boxf^{-1}$ and $R_\boxp = R_\diaf^{-1}$; therefore the clauses in
Figure~\ref{fig:semantics} are couched in terms of $R_\diaf$ 
and $R_\boxf$ only.
Our frame conditions F1$\diaf$ and F2$\boxf$ are also used by Simpson
whose F2 captures the ``persistence of being seen
by''~\cite[page~51]{simpson1994} while for us F2$\boxf$ is simply the
``persistence of $\diap$''.

$\lbikt$ is sound with respect to $\bikt$. The soundness proof
is straightforward by the definition of the semantics and
the inference rules.
\begin{theorem}[Soundness]
If $A$ is a $\bikt$-formula and $\emptyset \struct A$ is
$\lbikt$-derivable, then $A$ is $\bikt$-valid. 
\end{theorem}

We conjecture that $\dbikt$ is also complete w.r.t. the semantics.
We give an outline here:
The proof follows the usual counter-model construction
  technique for intuitionistic and tense logics; the non-trivial
  addition is showing that the resulting models satisfy the frame
  conditions F1$\diaf$
and F2$\boxf$. We will
  show that our propagation rules allow us to simulate the frame
  conditions. We do the case for F1$\diaf$; the other is
  similar.

We need to consider all the derivation fragments ending with a deep
sequent such that it contains the syntactic equivalent of three worlds
$x$, $y$ and $z$ such that $x \leq y$ and $x R_\diaf z$. Then we need
to show that there exists a $w$ such that $y R_\diaf w$ and $z \leq
w$. We will take $z$ to be the required $w$, which means we need to
show that $y R_\diaf z$ and $z \leq z$. The latter follows immediately
because $\leq$ is reflexive. The former demands the following:

\begin{enumerate}
	\item If $y \not\Vdash \diaf B$ then $z \not\Vdash B$
	\item If $z \Vdash \boxp A$ then $y \Vdash A$
\end{enumerate}

The following derivation fragments illustrate the required
propagations for (1) on the left and (2) on the right:

\begin{tabular}{cc}
\begin{minipage}{6.5cm}
$$
\AxiomC{$\Sigma[(\circwhite{Z \struct B} \struct U, \diaf B) \struct Y, \diaf B]$}
\RightLabel{$\diaf_{R2}$} \UnaryInfC{$\Sigma[(\circwhiteshort{Z} \struct U, \diaf B) \struct Y, \diaf B]$}
\RightLabel{$\struct_{R2}$} \UnaryInfC{$\Sigma[(\circwhiteshort{Z} \struct U) \struct Y, \diaf B]$}
\DisplayProof
$$

$$
\AxiomC{$\Sigma[\circwhite{Z \struct B} \struct (Y_1 \struct Y_2, \diaf B), \diaf B]$}
\RightLabel{$\diaf_{R2}$} \UnaryInfC{$\Sigma[\circwhiteshort{Z} \struct (Y_1 \struct Y_2, \diaf B), \diaf B]$}
\RightLabel{$\struct_{R1}$} \UnaryInfC{$\Sigma[\circwhiteshort{Z} \struct (Y_1 \struct Y_2, \diaf B)]$}
\DisplayProof
$$

$$
\AxiomC{$\Sigma[Z \struct \circblack{(Y_1 \struct Y_2, \diaf B), \diaf B}, B]$}
\RightLabel{$\diaf_{R1}$} \UnaryInfC{$\Sigma[Z \struct \circblack{(Y_1 \struct Y_2, \diaf B), \diaf B}]$}
\RightLabel{$\struct_{R1}$} \UnaryInfC{$\Sigma[Z \struct \circblack{Y_1 \struct Y_2, \diaf B}]$}
\DisplayProof
$$
\end{minipage}
&
\begin{minipage}{6.5cm}
$$
\AxiomC{$\Sigma[A, (\circwhite{\boxp A, Z} \struct X) \struct U]$}
\RightLabel{$\boxp_{L1}$} \UnaryInfC{$\Sigma[(\circwhite{\boxp A, Z} \struct X) \struct U]$}
\DisplayProof
$$

$$
\AxiomC{$\Sigma[A, \circwhite{\boxp A, Z} \struct (A, Y_1 \struct Y_2)]$}
\RightLabel{$\struct_{L2}$} \UnaryInfC{$\Sigma[A, \circwhite{\boxp A, Z} \struct (Y_1 \struct Y_2)]$}
\RightLabel{$\boxp_{L1}$} \UnaryInfC{$\Sigma[\circwhite{\boxp A, Z} \struct (Y_1 \struct Y_2)]$}
\DisplayProof
$$

$$
\AxiomC{$\Sigma[\boxp A \struct \circblack{A \struct (A, Y_1 \struct Y_2)}]$}
\RightLabel{$\struct_{L2}$} \UnaryInfC{$\Sigma[\boxp A \struct \circblack{A \struct (Y_1 \struct Y_2)}]$}
\RightLabel{$\boxp_{L2}$} \UnaryInfC{$\Sigma[\boxp A \struct \circblack{Y_1 \struct Y_2}]$}
\DisplayProof
$$
\end{minipage}
\end{tabular}

\section{Modularity, Extensions and Classicality}
\label{sec:ext}

We first exhibit the modularity of our deep calculus $\dbikt$
by showing that fragments of $\dbikt$ obtained by restricting
the language of formulae and structures also satisfy cut admissibility. 
For example, by allowing only $\lint$ formulae and structures, 
and restricting to rules that only affect those structures, 
we get a subsystem of $\dbikt$ that admits cut admissibility. 
We then show how we can obtain Ewald's intuitionistic tense logic
IKt~\cite{ewald1986}, Simpson's intuitionistic modal logic
IK~\cite{simpson1994} and regain classical tense logic \kt. We also
discuss extensions of $\dbikt$ with axioms $T$, $4$ and $B$
but they do not correspond semantically to reflexivity,
transitivity and symmetry~\cite{simpson1994}.

\paragraph{Modularity}

A nested sequent is {\em purely modal} if contains no
occurrences of $\bullet$ nor its formula translates $\boxp$ and $\diap.$
We write $\di$ for the sub-system of $\dbikt$ containing only 
the rules $id$, the logical rules for intuitionistic connectives, and 
the propagation rules for $\struct.$ The logical system $\dik$ is
obtained by adding to $\di$ the deep introduction rules for $\boxf$
and $\diaf$, and the propagation rules $\boxf_{L2}$
and $\diaf_{R2}.$ The logical system $\dbi$ is obtained by adding
to $\di$ the deep introduction rules for $\dimpl.$ 
In the following, we say that a formula is an {\em $\intk$-formula} 
if it is composed from propositional variables,
intuitionistic connectives, and $\boxf$ and $\diaf.$
Observe that in $\dbikt$, the only rules that create $\circblackshort$
upwards are $\diap_L$ and $\boxp_R$.  Thus in every
$\dbikt$-derivation $\pi$ of an $\intk$ formula, the internal sequents
in $\pi$ are purely modal, and hence $\pi$ is also a
$\dik$-derivation.  This observation gives immediately the following
modularity result. 

\begin{theorem}[Modularity]
\label{thm:modularity}
Let $A$ be an \lint{} (resp. $\biint$ and $\intk$) formula. 
The nested sequent $\emptyset \struct A$ is 
$\di$-derivable (resp. $\dbi$- and $\dik$-derivable) 
iff $\emptyset \struct A$ is $\dbikt$-derivable.
\end{theorem}
A consequence of Theorem~\ref{thm:cut-elim}, Theorem~\ref{thm:soundness},
Theorem~\ref{thm:completeness} and Theorem~\ref{thm:modularity},
is that the cut rule is admissible in $\di$, $\dbi$ and $\dik.$
As the semantics of $\lbikt$ (hence, also $\dbikt$) is
conservative w.r.t. to the semantics of both intuitionistic
and bi-intuitionistic logic, the following completeness
result holds.

\begin{theorem}
An \lint{} (resp.\ $\biint$) formula $A$ is
valid in \lint{} (resp.\ $\biint$) iff 
$\emptyset \struct A$ is derivable in $\di$ (resp. $\dbi$).
\end{theorem}

\paragraph{Obtaining Ewald's IKt}

\begin{figure}[t]
\mysmall{
$$
\AxiomC{$$}
\RightLabel{$id$} \UnaryInfC{$A \struct A$}
\RightLabel{$w_R$} \UnaryInfC{$A \struct A, \circblackshort{\diaf B}$}
\AxiomC{$$}
\RightLabel{$id$} \UnaryInfC{$B \struct B$}
\RightLabel{$\diaf_R$} \UnaryInfC{$\circwhiteshort{B} \struct \diaf B$}
\RightLabel{$rp_\circblackshort$} \UnaryInfC{$B \struct \circblackshort{\diaf B}$}
\RightLabel{$w_L$} \UnaryInfC{$B, A \struct \circblackshort{\diaf B}$}
\BinaryInfC{$A \impl B, A \struct \circblackshort{\diaf B}$}
\RightLabel{$\struct_R$} \UnaryInfC{$A \impl B \struct A \struct \circblackshort{\diaf B}$}
\RightLabel{$\boxf_L$} \UnaryInfC{$\boxf(A \impl B) \struct \circwhite{A \struct \circblackshort{\diaf B}}$}
\RightLabel{$rp_\circwhiteshort$} \UnaryInfC{$\circblack{\boxf(A \impl B)} \struct A \struct \circblackshort{\diaf B}$}
\RightLabel{$s_R$} \UnaryInfC{$A, \circblack{\boxf(A \impl B)} \struct \circblackshort{\diaf B}$}
\RightLabel{$\struct_R$} \UnaryInfC{$A \struct \circblack{\boxf(A \impl B)} \struct \circblackshort{\diaf B}$}
\RightLabel{$\bullet\struct_R$} \UnaryInfC{$A \struct \circblack{\boxf(A \impl B) \struct \diaf B}$}
\RightLabel{$rp_\circblackshort$} \UnaryInfC{$\circwhiteshort{A} \struct (\boxf(A \impl B) \struct \diaf B)$}
\RightLabel{$\diaf_L$} \UnaryInfC{$\diaf A \struct (\boxf(A \impl B) \struct \diaf B)$}
\RightLabel{$s_R$} \UnaryInfC{$\boxf(A \impl B), \diaf A \struct \diaf B$}
\RightLabel{$\impl_R \times 2$} \UnaryInfC{$\struct \boxf(A \impl B) \impl (\diaf A \impl \diaf B)$}
\DisplayProof
\qquad \qquad
\AxiomC{$$}
\RightLabel{$id$} \UnaryInfC{$A \struct A, \circblackshort{\bot}$}
\AxiomC{$$}
\RightLabel{$\bot_L$} \UnaryInfC{$\bot, A \struct \circblackshort{\bot}$}
\RightLabel{$\impl_L$} \BinaryInfC{$A \impl \bot, A \struct \circblackshort{\bot}$}
\RightLabel{$\struct_R$} \UnaryInfC{$A \impl \bot \struct A \struct \circblackshort{\bot}$}
\RightLabel{$\boxf_L$} \UnaryInfC{$\boxf(A \impl \bot) \struct \circwhite{A \struct \circblackshort{\bot}}$}
\RightLabel{$rp_\circwhiteshort$} \UnaryInfC{$\circblack{\boxf(A \impl \bot)} \struct A \struct \circblackshort{\bot}$}
\RightLabel{$s_R$} \UnaryInfC{$A, \circblack{\boxf(A \impl \bot)} \struct \circblackshort{\bot}$}
\RightLabel{$\struct_R$} \UnaryInfC{$A \struct \circblack{\boxf(A \impl \bot)} \struct \circblackshort{\bot}$}
\RightLabel{$\bullet\struct_R$} \UnaryInfC{$A \struct \circblack{\boxf(A \impl \bot) \struct \bot}$}
\RightLabel{$rp_\circblackshort$} \UnaryInfC{$\circwhiteshort{A} \struct \boxf(A \impl \bot) \struct \bot$}
\RightLabel{$\diaf_L$} \UnaryInfC{$\diaf A \struct \boxf(A \impl \bot) \struct \bot$}
\RightLabel{$s_R$} \UnaryInfC{$\boxf(A \impl \bot), \diaf A \struct \bot$}
\RightLabel{$\impl_R \times 2$} \UnaryInfC{$\struct \boxf(A \impl \bot) \impl (\diaf A \impl \bot)$}
\DisplayProof
$$
}
\caption{Derivations of Simpson's axiom 2 and Ewald's axiom 5 (left) and Ewald's axiom 7 (right)}
\label{fig:ewald1}
\end{figure}

\begin{figure}[t]
$$
\mysmall{
\AxiomC{$$}
\RightLabel{$id$} \UnaryInfC{$A \struct A$}
\RightLabel{$\diaf_{R}$} \UnaryInfC{$\circwhiteshort{A} \struct \diaf A$}
\AxiomC{$$}
\RightLabel{$id$} \UnaryInfC{$B \struct B$}
\RightLabel{$\boxf_{L}$} \UnaryInfC{$\boxf B \struct \circwhiteshort{B}$}
\RightLabel{$w_{L}$} \UnaryInfC{$\boxf B, \circwhiteshort{A} \struct \circwhiteshort{B}$}
\RightLabel{$\impl_L$} \BinaryInfC{$\diaf A \impl \boxf B, \circwhiteshort{A} \struct \circwhiteshort{B}$}
\RightLabel{$\struct_{R}$} \UnaryInfC{$\diaf A \impl \boxf B \struct \circwhiteshort{A} \struct \circwhiteshort{B}$}
\RightLabel{$\circ\struct_R$} \UnaryInfC{$\diaf A \impl \boxf B \struct \circwhite{A \struct B}$}
\RightLabel{$rp_\circ$} \UnaryInfC{$\circblack{\diaf A \impl \boxf B} \struct A \struct B$}
\RightLabel{$s_R$} \UnaryInfC{$\circblack{\diaf A \impl \boxf B}, A \struct B$}
\RightLabel{$\impl_R$} \UnaryInfC{$\circblack{\diaf A \impl \boxf B} \struct A \impl B$}
\RightLabel{$rp_\circ$} \UnaryInfC{$\diaf A \impl \boxf B \struct \circwhite{A \impl B}$}
\RightLabel{$\boxf_R$} \UnaryInfC{$\diaf A \impl \boxf B \struct \boxf(A \impl B)$}
\RightLabel{$\impl_R$} \UnaryInfC{$\struct (\diaf A \impl \boxf B) \impl \boxf(A \impl B)$}
\DisplayProof
\qquad
\AxiomC{$$}
\RightLabel{$id$} \UnaryInfC{$A \struct A$}
\RightLabel{$\boxp_{L}$} \UnaryInfC{$\boxp A \struct \circblackshort{A}$}
\RightLabel{$rp_{\circ}$} \UnaryInfC{$\circwhite{\boxp A} \struct A$}
\RightLabel{$w_{R}$} \UnaryInfC{$\circwhite{\boxp A} \struct A, \circwhite{\diap B}$}
\AxiomC{$$}
\RightLabel{$id$} \UnaryInfC{$B \struct B$}
\RightLabel{$\diap_{R}$} \UnaryInfC{$\circblackshort{B} \struct \diap B$}
\RightLabel{$rp_{\circ}$} \UnaryInfC{$B \struct \circwhite{\diap B}$}
\RightLabel{$w_{L}$} \UnaryInfC{$B, \circwhite{\boxp A} \struct \circwhite{\diap B}$}
\RightLabel{$\impl_{L}$} \BinaryInfC{$A \impl B, \circwhite{\boxp A} \struct \circwhite{\diap B}$}
\RightLabel{$\struct_R$} \UnaryInfC{$A \impl B \struct \circwhite{\boxp A} \struct \circwhite{\diap B}$}
\RightLabel{$\circ\struct_R$} \UnaryInfC{$A \impl B \struct \circwhite{\boxp A \struct \diap B}$}
\RightLabel{$rp_{\bullet}$} \UnaryInfC{$\circblack{A \impl B} \struct \boxp A \struct \diap B$}
\RightLabel{$s_{R}$} \UnaryInfC{$\circblack{A \impl B},  \boxp A \struct \diap B$}
\RightLabel{$\impl_{R}$} \UnaryInfC{$\circblack{A \impl B} \struct \boxp A \impl \diap B$}
\RightLabel{$\diap_{L}$} \UnaryInfC{$\diap(A \impl B) \struct \boxp A \impl \diap B$}
\RightLabel{$\impl_R$} \UnaryInfC{$\struct \diap(A \impl B) \impl (\boxp A \impl \diap B)$}
\DisplayProof
}
$$
\caption{Derivations of Simpson's axiom 5 and Ewald's axiom 10 (left) and Ewald's axiom 11' (right)}
\label{fig:ewald2}
\end{figure}

To obtain Ewald's IKt~\cite{ewald1986} 
we need to collapse $R_\diaf$ and $R_\boxf$ into one temporal 
relation $R$ and leave out our semantic clauses for 
$\dimpl$ and $\dneg$. That is, we need to add
the following conditions to the basic semantics: $R_\diaf \subseteq
R_\boxf$ and $R_\boxf \subseteq R_\diaf$. 
Proof theoretically, this is captured by extending $\lbikt$ with
the structural rules: 
$$
\AxiomC{$X \struct \circblackshort{Y} \struct \circblackshort{Z}$}
\RightLabel{$\bullet\struct_R$} \UnaryInfC{$X \struct \circblack{Y \struct Z}$}
\DisplayProof
\qquad
\AxiomC{$X \struct \circwhiteshort{Y} \struct \circwhiteshort{Z}$}
\RightLabel{$\circ\struct_R$} \UnaryInfC{$X \struct \circwhite{Y \struct Z}$}
\DisplayProof
$$
We refer to the extension of $\lbikt$ with these two structural rules
as $\lbikte.$

Simpson's intuitionistic modal logic IK~\cite{simpson1994} can then 
be obtained from Ewald's system by restricting the language
to the modal fragment. 
Note that cut-elimination still holds for $\lbikte$ 
because these structural rules are closed under formula
substitution and the cut-elimination proof for $\lbikt$ 
still goes through when additional structural rules of this kind 
are added. We refer the reader to \cite{gorepostniecetiu2009} for a discussion
on how cut elimination can be proved for this kind of extensions.

A $\bikt$-frame is an {\em $E$-frame} if
$R_\boxf = R_\diaf.$ A formula $A$ is {\em $E$-valid} if
it is true in all worlds of every $E$-model.
An IKt formula $A$ is a
theorem of IKt iff it is $E$-valid~\cite{ewald1986}. 
The rules $\circ\struct_R$
and $\bullet\struct_R$ are sound for $E$-frames. 

\begin{lemma}\label{lm:collapse-rdia-rbox}
Rule $\circ\struct_R$ is sound iff $R_\boxf \subseteq R_\diaf$.
\end{lemma}
\begin{proof}
($\Leftarrow$) We show that if the frame condition holds, then the
  rule is sound. We assume that: (1)~$R_\boxf \subseteq R_\diaf$, and
  (2)~that the formula translation $\diaf A \impl \boxf B$ of the premise
  is valid. We then show that the formula translation $\boxf (A \impl
  B)$ of the conclusion is valid. 
For a contradiction, suppose that $\boxf (A \impl B)$ is not valid.
That is, there exists a world $u$ such that $u \not\Vdash \boxf (p \impl q)$. Then
(4) there exist worlds $x$ and $y$ such
that $u \leq x \ \&\ x R_{\boxf} y$ and $y \not\Vdash p \impl q$.
Thus there exists 
$z$ s.t. $z \geq y$ and $z \Vdash p$ and
$z \not\Vdash q$. 
The pattern $x R_{\boxf} y \leq z$ implies
there is a world $w$ with $x \leq wR_{\boxf} z$ by F2$\boxf$.
The frame condition (1) then gives $w R_{\diaf} z$ too,
meaning that $w \Vdash \diaf p$.
From (2) we get $w \Vdash \boxf q$, which 
gives us $z \Vdash q$, giving us the contradiction we seek.
Therefore the premise $\boxf (A
\impl B)$ is valid and the rule is sound.

($\Rightarrow$) We show that if the rule is sound, then the failure of
the frame condition gives a contradiction. So suppose that the rule is
sound. The rule implies that $\struct (\diaf A \impl \boxf B)
\impl \boxf (A \impl B)$ is derivable.  For a contradiction, suppose
we have a frame with $R_\boxf \not\subseteq R_\diaf$. That is, (5):
there exist $x$ and $y$ such that $x R_\boxf y$ but not $x R_\diaf
y$. Let $W = \{u, w, x, y, z \}$,  let $<$ be the relation $\{ (u, x), (x, w), (y, z) \}$
and let $\leq$ be the reflexive-transitive closure of $<.$
Let $R_\diaf ~~= \{ \}$, $R_\boxf ~~= \{ (x, y), (w, z) \}$	
and let $V(p) ~~= \{z\}, V(q) = \{\}.$
Then the model $\langle W, \leq$, $R_\diaf, R_\boxf, V \rangle$ satisfies
(5), and has $u \Vdash \diaf p \impl \boxf q$ but $u \not\Vdash \boxf
(p \impl q)$.
\end{proof}

\begin{lemma}
\label{lm:collapse-rdia-rbox2}
Rule $\bullet\struct_R$ is sound iff $R_\diaf \subseteq R_\boxf$.
\end{lemma}
\begin{proof}
$R_\diaf \subseteq R_\boxf$ means $R_\boxp \subseteq R_\diap$; 
the rest of the proof is analogous to the proof of Lemma~\ref{lm:collapse-rdia-rbox}.
\end{proof}

\begin{theorem}
\label{thm:ewald-sound}
If $A$ is derivable in $\lbikte$ then $A$ is $E$-valid.
\end{theorem}
\begin{proof}
Straightforward from the soundness of $\lbikt$ w.r.t.
$\bikt$-semantics (which subsumes Ewald's semantics)
and Lemma~\ref{lm:collapse-rdia-rbox} and Lemma~\ref{lm:collapse-rdia-rbox2}.
\end{proof}
Completeness of $\lbikte$ w.r.t. IKt and IK
can be shown by deriving the axioms of IKt and IK.
\begin{theorem}
\label{thm:ewald-complete}
System $\lbikte$ is complete w.r.t. Ewald's IKt and Simpson's IK.
\end{theorem}
\begin{proof}
We show the non-trivial cases; the rest are similar or easier. 
Derivations of Simpson's axiom 2 and Ewald's axiom 5 and 7 are given in Figure~\ref{fig:ewald1},
derivations of Simpson's axiom 5 and Ewald's axiom 10 and 11' are given in Figure~\ref{fig:ewald2}.
\end{proof}

\begin{theorem}[Conservativity over IKt and IK]
If $A$ is an IKt-formula (IK formula),
then $A$ is IKt-valid (IK-valid) iff
$\emptyset \struct A$ is derivable in $\lbikte.$
\end{theorem}

\paragraph{Regaining classical tense logic \kt}

To collapse \bikt{} to classical tense logic \kt{} we add the rules
$\circblackshort\struct_R$ and $\circwhiteshort\struct_L$, giving
Ewald's IKt with $R_{\diaf} = R_{\boxf}$ via
Lemmas~\ref{lm:collapse-rdia-rbox}-\ref{lm:collapse-rdia-rbox2}, and then add
following two rules:
$$
\AxiomC{$X_1, X_2 \struct Y_1, Y_2$}
\RightLabel{$s_L^{-1}$} 
\UnaryInfC{$(X_1 \struct Y_1), X_2 \struct Y_2$}
\DisplayProof
\qquad
\AxiomC{$X_1, X_2 \struct Y_1, Y_2$}
\RightLabel{$s_R^{-1}$} 
\UnaryInfC{$X_1 \struct Y_1, (X_2 \struct Y_2)$}
\DisplayProof
$$
The law of the excluded middle and the law of
(dual-)contradiction can then be derived as shown below:
$$
\AxiomC{$p \struct p, \bot$}
\RightLabel{$s^{-1}_L$}
\UnaryInfC{$(\emptyset \struct p), p \struct \bot$}
\RightLabel{$\impl_R$}
\UnaryInfC{$(\emptyset \struct p) \struct (p \impl \bot)$}
\RightLabel{$s_L$}
\UnaryInfC{$\emptyset \struct p, (p \impl \bot)$}
\RightLabel{$\lor_L$}
\UnaryInfC{$\emptyset \struct p \lor (p \impl \bot)$}
\DisplayProof
\qquad
\AxiomC{$p, \top \struct p$}
\RightLabel{$s^{-1}_R$}
\UnaryInfC{$\top \struct p, (p \struct \emptyset)$}
\RightLabel{$\dimpl_L$}
\UnaryInfC{$(\top \dimpl p) \struct (p \struct \emptyset)$}
\RightLabel{$s_R$}
\UnaryInfC{$p, (\top \dimpl p) \struct \emptyset$}
\RightLabel{$\land_R$}
\UnaryInfC{$p \land (\top \dimpl p) \struct \emptyset$}
\DisplayProof
$$

\paragraph{Further extensions}

Our previous work on deep inference systems for classical
tense logic~\cite{gorepostniecetiu2009} shows that extensions of
classical tense logic with some standard modal axioms can be formalised
by adding numerous propagation rules to the deep inference system
for classical tense logic given in that paper. 
We illustrate here with a few examples how such an approach to extensions
with modal axioms can be applied to $\bikt$. 
Figure~\ref{fig:ext} shows the propagation rules that are needed
to derive axiom T, 4 and B. For each rule, the derivation of
the corresponding axiom is given below the rule. 
Other nesting combinations will be needed for full 
completeness. Dual rules allow derivations of $p \impl \diaf p$
and $\diaf\diaf p \impl \diaf p$.
The complete treatement of these and other possible extensions
of $\lbikt$ is left for future work.

\begin{figure}[t]
$$
\mysmall{
\begin{array}{ccc}
\mbox{
\AxiomC{$\Sigma^-[A, \boxf A]$}
\RightLabel{$T\boxf$}
\UnaryInfC{$\Sigma^-[\boxf A]$}
\DisplayProof
}
&
\mbox{
\AxiomC{$\Sigma[\boxf A, X \struct \circwhite{\boxf A \struct Y}, Z]$}
\RightLabel{$4\boxf_{L}$}
\UnaryInfC{$\Sigma[\boxf A, X \struct \circwhiteshort{Y}, Z]$}
\DisplayProof
}
&
\mbox{
\AxiomC{$\Sigma^-[A, \circwhite{\boxf A, X}]$}
\RightLabel{$B\boxf_{L}$}
\UnaryInfC{$\Sigma^-[\circwhite{\boxf A, X}]$}
\DisplayProof
} \\ \\
\mbox{
\AxiomC{}
\RightLabel{$id$}
\UnaryInfC{$p, \boxf p \struct p$}
\RightLabel{$T\boxf$}
\UnaryInfC{$\boxf p \struct p$}
\RightLabel{$\impl_R$}
\UnaryInfC{$\struct~\boxf p \impl p$}
\DisplayProof
}
&
\mbox{
\AxiomC{}
\RightLabel{$id$}
\UnaryInfC{$\boxf p \struct \circwhite{\boxf p \struct \boxf p}$}
\RightLabel{$4\boxf_{L}$}
\UnaryInfC{$\boxf p \struct \circwhiteshort\boxf p$}
\RightLabel{$\boxf_R$}
\UnaryInfC{$\boxf p \struct \boxf\boxf p$}
\RightLabel{$\impl_R$}
\UnaryInfC{$\struct~\boxf p \impl \boxf\boxf p$}
\DisplayProof
}
&
\mbox{
\AxiomC{}
\RightLabel{$id$}
\UnaryInfC{$p, \circwhiteshort\boxf p \struct  p$}
\RightLabel{$B\boxf_L$}
\UnaryInfC{$\circwhiteshort\boxf p \struct  p$}
\RightLabel{$\diaf_L$}
\UnaryInfC{$\diaf\boxf p \struct  p$}
\RightLabel{$\impl_R$}
\UnaryInfC{$\struct~\diaf\boxf p \impl  p$}
\DisplayProof
}
\end{array}
}
$$
\caption{Some example propagation rules and the axioms they capture}
\label{fig:ext}
\end{figure}

\bibliographystyle{aiml10}

\end{document}